\documentclass{article}
\usepackage{subcaption}
\usepackage{graphicx}
\usepackage{float}
\usepackage{amsmath}
\usepackage{amssymb}
\usepackage{amsthm}
\usepackage{multirow}
\usepackage{geometry}
\usepackage{cite,color}

\usepackage[ruled,vlined]{algorithm2e}
\newtheorem{thm}{Theorem}[section]
\newtheorem{lem}[thm]{Lemma}

\newtheorem{defn}{Definition}[section]

\newtheorem{assu}{Assumption}[section]
\newtheorem{rk}{Remark}[section]
\numberwithin{equation}{section}

\begin{document}
\title{CT Image Reconstruction by Spatial-Radon Domain Data-Driven Tight Frame Regularization}
\author{Ruohan Zhan\footnotemark[1] \and
Bin Dong\footnotemark[2]}

\footnotetext[1]{Yuanpei College, Peking University, Beijing, CHINA (zrhan@pku.edu.cn).}
\footnotetext[2]{Corresponding author. Beijing International Center for Mathematical
Research, Peking University, Beijing, CHINA (dongbin@math.pku.edu.cn).
Research supported in part by the Thousand Talents Plan of China.}

\maketitle
\begin{abstract}
This paper proposes a spatial-Radon domain CT image reconstruction model based on data-driven tight frames (SRD-DDTF). The proposed SRD-DDTF model combines the idea of joint image and Radon domain inpainting model of \cite{Dong2013X} and that of the data-driven tight frames for image denoising \cite{cai2014data}. It is different from existing models in that both CT image and its corresponding high quality projection image are reconstructed simultaneously using sparsity priors by tight frames that are adaptively learned from the data to provide optimal sparse approximations. An alternative minimization algorithm is designed to solve the proposed model which is nonsmooth and nonconvex. Convergence analysis of the algorithm is provided. Numerical experiments showed that the SRD-DDTF model is superior to the model by \cite{Dong2013X} especially in recovering some subtle structures in the images.

\end{abstract}

\section{Introduction}

X-ray computed tomography (CT) has been widely used in clinic due to its great ability in visualizing interior structures. However, additional imaging dose to patients' healthy radiosensitive cells or organs has always been a serious clinical concern \cite{Islam2006Patient,brenner2007computed,Ding2009Radiation}. Low-dose CT is highly desirable if satisfactory image quality can be maintained for a specific clinical task. One commonly adopted strategy to achieve low-dose CT imaging, especially for cone beam CT (CBCT), is to reduce the total number of projections. However, this may also lead to degraded restored images if the reconstruction algorithm is not properly designed to incorporate missing information due to incomplete angular sampling and system noise. Therefore, many classical algorithms based on a complete angular sampling such as filtered back projection (FBP) \cite{Katsevich2002THEORETICALLY} will generate undesirable artifacts due to lack of measurements. Other inversion techniques such as pseudo-inverse based methods \cite{Antsaklis1991Stability,Tikhonov1978Solutions} also perform poorly at the presence of noise. Therefore, a more effective and robust method is needed to achieve satisfactory reconstruction for clinical purposes. In this paper, we shall focus on the problem of low-dose planer fan beam CT reconstruction of 2D images. However, the same modeling concept can be easily applied to 3D CBCT image reconstruction.

Assume that the X-ray point source with a fixed milliampere-second (mAs) setting rotates along a circle centered at the object, and a linear detector array is used. CT image reconstruction can be casted as the following linear inverse problem in discrete setting \begin{equation}\label{pro}
  f=Pu+\epsilon,
\end{equation}
where $P$ is the projection matrix generated by the Sidden's Algorithm \cite{Siddon1985Fast}, $f$ is the projection image whose rows indicate the data collected by each detector and columns indicate data collected from different projection angles, and $\epsilon$ is additive Gaussian white noise. To reduce radiation dose, one common way is to reduce the number of projection angles which leads to an under-determined (or rank deficient) linear system (or matrix $P$). This is the main challenge of reconstructing a desirable CT image $u$ from its projections $f$ via \eqref{pro}, and also the reason why traditional CT reconstruction algorithms such as FBP and pseudo-inverse based methods do not perform well.

In image restoration, many problems can be formulated as the same linear inverse problem \eqref{pro} with $P$ taking different forms for different image restoration problems. For instance, $P$ is an identity operator for image denoising; a convolution operator for image deblurring; and a restriction operator for image inpainting. Image restoration has the same challenge as CT image reconstruction, which is the rank deficiency of the matrix $P$. A good image restoration method should be capable of smoothing the image so that noise and other artifacts are suppressed to the greatest extend, while at the same time, preserving important image features such as edges. This is a challenging task since smoothing and preservation of features are often contradictory to each other.

Most of the existing models and algorithms for image restoration are transformation based. The key to the success of any transformation based image restoration method is to find a transform that can identify local features from the given image, or in other words, to separate singularities and smooth image components. Such property is closely related to the concept called ``sparse approximation" which is broadly adopted in various problems in image processing, image analysis, matrix completion, high dimensional data analysis, etc. Interested reader should consult \cite{Dong2015} for a review of classical and recent developments of image restoration methods.

One of the most successful transformations in image restoration is the wavelet frame transform. It has been implemented with excellent results in both classical \cite{cai2008framelet,chan2003wavelet,cai2009split,cai2009linearized,zhang2011l0min,dong2011efficient,liang2013wavelet,hou2013recovering} and some more challenging image restoration problems \cite{cai2009blind,cai2009blindcvpr,dong2012wavelet,gong2014image}. Frames provide vast flexibility in designing adaptive and non-local filters with improved performance in applications \cite{quan2014data,cai2014data,bao2014convergence,tai2015adaptive}. The application of wavelet frames has gone beyond image restoration. They have been successfully used in video processing \cite{ji2011robust}, image segmentation \cite{Dong2010Seg,tai2013wavelet} and classifications \cite{wendt2009wavelet,bao2014l0}. More recently, wavelet frames are constructed on non-flat domains such as surfaces \cite{jiang2011highly,dong2015surf} and graphes \cite{hammond2011wavelets,gavish2010multiscale,leonardi2013tight,dong2015sparse} with applications to denoising \cite{jiang2011highly,dong2015surf,dong2015sparse} and classifications \cite{dong2015sparse}.

Another class of methods for image restoration that have been developed through a rather different path is the PDE based approach \cite{sapiro2001geometric,OF,ChanShen} which started with the refined total variation (TV) model \cite{ROF} and anisotropic diffusion \cite{PM}. The PDE based approach includes variational and (nonlinear) PDE based methods. Both variational and PDE methods can be understood as transformation based methods as well, where the transformations are the differential operators involved in the models \cite{DJS2013}.

In recent work by \cite{CDOS2011,CDS2014}, fundamental connections between wavelet frame based approach and variational methods were established. In particular, connections to the total variation model \cite{ROF} was established in \cite{CDOS2011}, and to the Mumford-Shah model \cite{mumford1989optimal} was established in \cite{CDS2014}. Furthermore, in \cite{DJS2013}, the authors established a generic connection between iterative wavelet frame shrinkage and general nonlinear evolution PDEs which include the Perona-Malik equation \cite{PM} and the shock-filters \cite{Osher1990} as special cases. The series of three papers \cite{CDOS2011,DJS2013,CDS2014} showed that wavelet frame transforms are discretization of differential operators in both variational and PDE frameworks, and such discretization is superior to some of the traditional finite difference schemes for image restoration. This new understanding essentially merged the two seemingly unrelated areas: wavelet frame base approach and PDE based approach. It also gave birth to many innovative and more effective image restoration models and algorithms.

The concept of sparse approximation via linear transformations originated from image restoration was also applied to CT image reconstruction due to the similarity of the problems in nature. For example, wavelet frame based methods are developed for standard CT image reconstruction \cite{jia2011gpu}, for 4D CT image reconstruction \cite{gao2011robust,gao20124d,cai2014cine} and spectral CT reconstruction \cite{zhao2013tight}. TV-based regularization model was also applied to CT image reconstruction in \cite{sidky2006accurate,sidky2008image,Jia2010GPU,zhang2010bregmanized,chen2013limited,han2013image}. Many other regularization based methods for CT image reconstruction have also been introduced \cite{thibault2007three,wang2009iterative,tang2009performance,ouyang2011effects,lee2012improved,lauzier2013characterization}, as well as dictionary learning based methods \cite{xu2012low,chen2013improving,zhou2013adaptive,chen2014artifact}.

However, all methods mentioned above attempted to recover a good CT image $u$ with a fixed projection image $f$. Various sparsity based prior knowledge on the CT image $u$ have been used, while the prior knowledge on $f$ is yet to be fully exploited. The projection image $f$ we collect using under-sampled angles will suffer from lack of angular resolution and measurement noise. Therefore, to reconstruct a high quality CT image $u$ from \eqref{pro}, we need to restore a high quality (improved angular resolution and reduced noise) projection image $f$ using properly chosen prior knowledge on $f$. Since $f$ and $u$ are linked by the linear inverse problem \eqref{pro}, it is more effective to restore both $u$ and $f$ simultaneously. Such modeling philosophy was first introduced in \cite{Dong2013X} with success, where the authors proposed the following optimization model based on sparse approximation of tight wavelet frames
\begin{equation}\label{primal}
\min_{f,u}\ \frac{1}{2}\|R_{\Lambda^c}(Pu-f)\|_2^2+\frac{1}{2}\|R_\Lambda(Pu)-f_0\|_2^2+\frac{\kappa}{2}\|R_\Lambda f-f_0\|_2^2+\lambda_1\|W_1f\|_1+\lambda_2\|W_2 u\|_1.
\end{equation}
Here, $f_0$ is the projection image we collect from the scanner defined on the grid $\Lambda$ of size $N_D\times N_P$, where $N_D$ is the total number of detectors and $N_P$ is the number of angular projections. The projection image $f$ that \eqref{primal} tries to recover is defined on a grid $\Omega\supset\Lambda$ of size $N_D\times \widetilde N_P$. In this paper, we focus on the case $\widetilde N_P=2N_P$, which means we want to recover a projection image $f$ that has twice the angular resolution as that of $f_0$. The operator $R_\Lambda$ is the restriction operator associated to the set $\Lambda$. The first three terms of \eqref{primal} makes sure that $f$ is consistent with $f_0$ on $\Lambda$ and $Pu\approx f$. The last two terms are the sparsity priors assumed on $u$ and $f$, where $W_1$ and $W_2$ are two (possibly different) tight wavelet frame transforms. We refer the interested readers to \cite{Dong2013X} for more details.

Although positive results were reported in \cite{Dong2013X}, the sparsity priors based on $W_1$ and $W_2$ can be further improved. It is known in the literature of image restoration that wavelet frames can sparsely approximate images or piecewise smooth functions in general. However, for a specifically given image, the sparse approximation by a pre-constructed wavelet frame system may not be ideal. This is the main reason why data-driven tight frames or bi-frames generally outperforms regular wavelet frames in image restoration \cite{cai2014data,bao2014convergence,tai2015adaptive}. In this paper, we propose to use data-driven tight frames of \cite{cai2014data,bao2014convergence} as our sparsity priors for both $u$ and $f$. The contribution of this paper is threefold: (1) the introduction of a spatial-Radon domain CT image reconstruction model based on data-driven tight frames (SRD-DDTF); (2) the design of an alternative optimization algorithm; and (3) convergence analysis of the proposed algorithm.

The rest of the paper is organized as follows. In Section \ref{Sec:Review} we review the basic knowledge of wavelet frames and data-driven tight frames. In Section \ref{Sec:Model-Alg}, we introduce our spatial-Radon domain CT image reconstruction model based on data-driven tight frames, followed by an efficient algorithm and its convergence analysis. In Section \ref{Sec:Numerical}, we present some numerical simulations, and the concluding remarks are given in Section \ref{Sec:Conclusion} at the end.

\section{Reviews and Preliminaries}\label{Sec:Review}

\subsection{Tight Wavelet Frames}

In this section, we briefly introduce the concept of tight wavelet frames. The interested readers should consult \cite{ron1997affine,ron1997affineII,Dau,Daubechies2003} for theories of frames and wavelet frames, \cite{ShenICM2010,Dong2015} for a short survey on the theory and applications of frames, and \cite{Dong2010IASNotes} for a more detailed survey.

For a given set of functions $\Psi=\{\psi_1,\psi_2,\dots,\psi_r\}\subset L_2(\mathbb{R})$, the quasi-affine wavelet system is defined as
\begin{equation*}\label{E:Quasi:Affine}
X(\Psi)=\{\psi_{j,n,k}:\ 1\le j \le r; n\in\mathbb{Z},k\in\mathbb{Z}\},
\end{equation*}
where $\psi_{j,n,k}$ is defined by
\begin{equation*}\label{E:Quasi:Affine:1}
\psi_{j,n,k}:=\left\{\begin{array}{cc}
2^{\frac{n}2}\psi_{j}(2^n\cdot-k),&n\ge 0;\\
2^{n}\psi_{j}(2^n\cdot-2^{n}k),&n<0.
\end{array}\right.
\end{equation*}
The system $X(\Psi)$ is called a \emph{tight wavelet frame} of $L_2(\mathbb{R})$ if
\begin{equation*}
  f=\sum_{g\in X(\Psi) }\langle f,g\rangle g
\end{equation*}
holds for all $f\in L_2(\mathbb{R})$, where $\langle\cdot,\cdot\rangle$ is the inner product in $L_2(\mathbb{R})$. When $X(\Psi)$ forms a tight frame of $L_2(\mathbb{R})$, each function $\psi_{j}$, $j=1,\ldots, r$, is called a (tight) framelet and the whole system $X(\Psi)$ is called a tight wavelet frame.

The constructions of compactly supported and desirably (anti)symmetric framelets $\Psi$ are usually based on the multiresolution analysis (MRA) generated by some refinable function $\phi$ with refinement mask $a_0$ satisfying
\begin{equation*}\label{phimaskphi}
\phi=2\sum_{k\in\mathbb{Z}}{a_0[k]\phi(2\cdot-k)}.
\end{equation*}
The idea of an MRA-based construction of framelets $\Psi=\{\psi_1,\ldots, \psi_r\}$ is to find masks $a_j$, which are finite sequences (or filters), such that
\begin{equation}\label{psimaskphi}
\psi_{j} =
2\sum_{k\in\mathbb{Z}}{a_{j}[k]\phi(2\cdot-k)},\quad j=1,2,\ldots,r.
\end{equation}
The sequences $a_1,\ldots,a_r$ are called wavelet frame masks, or the high pass filters associated to the tight wavelet frame system, and $a_0$ is also known as the low pass filter.

The unitary extension principle (UEP) \cite{ron1997affine} provides a rather general characterization of MRA-based tight wavelet frames. Roughly speaking, as long as $\{a_1,\ldots,a_r\}$ are finitely supported and their Fourier series $\widehat a_j$ satisfy
\begin{equation}\label{UEPCondition}
\sum_{j=0}^{r}|\widehat{a}_{j}(\xi)|^2=1
\quad\text{and}\quad
\sum_{j=0}^{r}\widehat{a}_{j}(\xi)\overline{\widehat{a}_{j}(\xi+\pi)}=0,
\end{equation}
for all $\xi\in[-\pi,\pi]$, the quasi-affine system $X(\Psi)$ with $\Psi=\{\psi_1,\ldots, \psi_r\}$ defined by \eqref{psimaskphi} forms a tight frame of $L_2(\mathbb{R})$. Note that, some filters used in image restoration, such as those constructed in \cite{cai2014data,bao2014convergence} and some filter banks in \cite{DJS2013}, only satisfy the first condition of \eqref{UEPCondition}. In this case, the wavelet systems associated to these filter banks are not tight frames of $L_2(\mathbb{R})$ in general. However, these filter banks form tight frames for sequence space $\ell_2(\mathbb{Z})$ instead, which is sufficient for many image restoration problems.

In discrete setting, we denote $W$ as the fast decomposition transform and its adjoint $W^\top$ as the fast reconstruction transform. Both $W$ and $W^\top$ are formed by convolution operators with kernels $\{a_j\}_{j=0}^m$. Let $a$ be a filter in $\ell_2(\mathbb{Z})$. The convolution operator
$\mathcal{S}_a:\ell_2(\mathbb{Z})\rightarrow\ell_2(\mathbb{Z})$ associated to kernel $a$ is defined by
\begin{equation*}\label{}
[\mathcal{S}_au](n):=[a\ast u](n)=\sum_{k\in\mathbb{Z}}a(n-k)u(k).
\end{equation*}
Given a set of filters $\{a_j\}_{j=0}^m$, the associated analysis operator $W$ and its adjoint $W^\top$ are given by
\begin{equation}\label{E:Structure:W}
\begin{aligned}
& W=[\mathcal{S}^\top_{a_0(-\cdot)},\mathcal{S}^\top_{a_1(-\cdot)},\dots,\mathcal{S}^\top_{a_m(-\cdot)}]^\top,\\
& W^\top=[\mathcal{S}_{a_0},\mathcal{S}_{a_1},\dots,\mathcal{S}_{a_m}].
\end{aligned}
\end{equation}
It is not hard to verify that the filters $\{a_i\}_{j=0}^m$ satisfy the first condition of \eqref{UEPCondition} if and only if
\begin{equation}\label{uep2}
W^\top W=I.
\end{equation}

\subsection{Data-Driven Tight Frames}

This subsection is to briefly review the data-driven tight frames. Interested readers should refer to \cite{cai2014data,bao2014convergence} for details.

To learn a good tight frame $W$, taking the form of \eqref{E:Structure:W}, for a given image $u$, we solve the following optimization problem:
\begin{equation}\label{l0}
  \min_{v,W} \quad\lambda^2\|v\|_0+\|Wu-v\|_2^2,\quad W^\top W=I,
\end{equation}
where $\|\cdot\|_0$ is the $\ell_0$-``norm" that returns the number of non-zero entries of the input vector.

To solve \eqref{l0}, let us start with reformulating \eqref{l0}. Reshape all $N\times N$ patches of $u$ into vectors and put them together as column vectors of the matrix $G\in \mathbb{R}^{N^2\times p}$, where $p$ is the total number of patches. We put the filters $\{a_j\}_{j=0}^m$ associated to $W$ as column vectors of the matrix $D\in \mathbb{R}^{N^2\times m}$. For simplicity, we focus on the case $m=N^2$, i.e. $D\in\mathbb{R}^{N^2\times N^2}$. Denote $V\in\mathbb{R}^{N^2\times p}$ as the tight frame coefficients. So, we have
\begin{equation}\label{reformulate}
\begin{aligned}
  &G=(g_1,g_2,\dots,g_p)\in\mathbb{R}^{N^2\times p},\\
  &D=(a_1,a_2,\dots,a_{N^2})\in\mathbb{R}^{N^2\times N^2},\\
  &V=(v_1,v_2,\dots,v_p)\in\mathbb{R}^{N^2\times p}.
\end{aligned}
\end{equation}
Thus, the decomposition operation can be written as $V=D^TG$, and the reconstruction operation can be written as $\tilde{G}=DV$. The condition $W^\top W=I$ is satisfied whenever $DD^\top=I$.  Now, we rewrite \eqref{l0} as
\begin{equation}\label{dict}
  \begin{aligned}
  &\min_{V,D} &\lambda^2\|V\|_0+\|D^TG-V\|_2^2, \quad D D^\top=I.
  \end{aligned}
\end{equation}

In \cite{cai2014data,bao2014convergence}, an alternative optimization algorithm was proposed to solve the problem \eqref{dict} and its convergence analysis was later given in \cite{bao2014convergence}. What makes the algorithm efficient is that both the subproblems for solving $D$ and $V$ respectively have closed-form solutions that can be efficiently computed. This algorithm can be written as
\begin{equation}\label{learndict}
  \begin{aligned}
  & D^{k+1}=XY^\top,\\
  & V^{k+1}=\mathcal{T}_{\lambda}((D^{k+1})^TG),
  \end{aligned}
\end{equation}
where $X$ and $Y$ are obtained by taking SVD of $G(V^k)^\top$, i.e. $G(V^k)^\top=X\Sigma Y^\top$, and $\mathcal{T}_{\lambda}$ is the hard-thresholding operator defined by
\begin{equation}\label{D:HT}
  \left(\mathcal{T}_{\lambda}(V)\right)[i,j]=\left\{
  \begin{aligned}
  &0,\quad \text{if }|x|<\lambda,\\
  &\{0,V[i,j]\},\quad \text{if }|x|=\lambda,\\
  &V[i,j],\quad \text{otherwise.}
  \end{aligned}\right.
\end{equation}

\section{Models and Algorithms}\label{Sec:Model-Alg}

\subsection{CT Image Reconstruction Model}

We first introduce some basic notation. Denote $P_0$ as the projection operator (computed using Sidden's algorithm \cite{Siddon1985Fast}) with $N_P$ projections and $N_D$ detectors, and $f_0$ as the observed projection image. Suppose $f_0$ is supported on the grid $\Lambda$ of size $N_D\times N_P$, with each pixel value representing the data received from each detector at each projection angle. Given $f_0$, our objective is to reconstruct a projection image $f$ with less noise and higher angular resolution than $f_0$, together with its corresponding high quality CT image $u$ at the same time. Let $f$ be supported on the grid $\Omega\supset\Lambda$ of size $N_D\times \widetilde N_P$ with $\widetilde N_p>N_p$. For simplicity, we focus on the case $\widetilde N_p=2N_p$, i.e. we want to restore an $f$ from $f_0$ with doubled angular resolution.

To ensure a high quality reconstruction of both $u$ and $f$, we shall enforce sparsity based regularization on both of the variables. In \cite{Dong2013X}, sparsity regularization based on tight wavelet frames was used and their numerical experiments showed the advantage of recovering both $u$ and $f$ simultaneously over the classical approach where $f$ is fixed, i.e. setting $f=f_0$. In this paper, instead of using a pre-constructed system as sparse approximation to $u$ and $f$, we adopt the idea of data-driven tight frames of \cite{cai2014data} to actively learn the optimal sparse representation for $u$ and $f$ based on the given data $f_0$. Our spatial-Radon domain CT image reconstruction model based on data-driven tight frames (SRD-DDTF) reads as follows:
\begin{equation}\label{obj}
\begin{aligned}
  &\min_{f,u,v_1,W_1,v_2,W_2}\ \frac{1}{2}\|R_{\Lambda^C}(Pu-f)\|_2^2+\frac{1}{2}\|R_\Lambda Pu-f_0\|_2^2+\frac{\kappa}{2}\|R_\Lambda f-f_0\|_2^2\\
  &\hspace*{1.5in}+\lambda_1\|v_1\|_0+\frac{\mu_1}{2}\|W_1f-v_1\|_2^2+\lambda_2\|v_2\|_0+\frac{\mu_2}{2}\|W_2 u-v_2\|_2^2,\\
  &\hspace*{0.35in}\text{s.t. }\quad W_i^\top W_i=I,\ i=1,2.
\end{aligned}
\end{equation}
where $R_{\Lambda^C}$ denotes the restriction on $\Omega\setminus{\Lambda}$, and $R_{\Lambda}$ denotes the restriction on $\Lambda$.

The first two terms $\frac{1}{2}\|R_{\Lambda^C}(Pu-f)\|_2^2+\frac{1}{2}\|R_\Lambda Pu-f_0\|_2^2$ is to ensure that $Pu\approx f$ on $\Lambda^C$ and $Pu\approx f_0$ on $\Lambda$, while the third term $\frac{\kappa}{2}\|R_\Lambda f-f_0\|_2^2$ is to ensure with restriction on $\Lambda$, $f\approx f_0$. The reason that we are not using the simpler fidelity term $\frac12\|Pu-f\|_2^2$ to enforce $Pu\approx f$ is because $f$ is the estimated projection data which may not be as reliable as $f_0$ on $\Lambda$. Therefore, in the domain $\Lambda$ where the actual projection image $f_0$ is available, we should make sure that $R_\Lambda Pu\approx f_0$.

The transforms $W_1$ and $W_2$ are tight frames (due to the constraints $W_i^\top W_i=I, i=1,2$), with frame coefficients $v_1$ and $v_2$, that are learned from $u$ and $f$ respectively. The use of the $\ell_0$-``norm" is to enforce sparsity of $v_1$ and $v_2$ which in turn grants sparse approximation to $u$ and $f$ by the transforms $W_1$ and $W_2$. The special structure of $W_i$ given by \eqref{E:Structure:W} and the constraints $W_i^\top W_i=I$ make the dictionary learning component of \eqref{obj} different from the popular K-SVD method \cite{aharon2006k}, where neither of the aforementioned properties is guaranteed to be satisfied. These properties make the learning of $W_i$ much faster than the K-SVD method, because the size of the problem is much smaller, while the performance is still comparable to the K-SVD method. Another drawback of the K-SVD method is that the learned dictionary is not guaranteed to be complete in the underlying Euclidean space, i.e. $W_i^\top W_i\ne I$. We refer the interested readers to \cite{cai2014data} for more details on the comparison between data-driven tight frames and the K-SVD method.

\subsection{Alternative Optimization Algorithms}

Given a projected data $f_0$, we first solve the following analysis based model \cite{cai2009split,elad2005simultaneous,
starck2005image}
\begin{equation}\label{anwave}
  \min_u \frac{1}{2}\|P_0u-f_0\|_2^2+\lambda\|Wu\|_1
\end{equation}
to obtain an initial reconstruction $u^0$. Then, we let $f^0=Pu^0$ to be the initial estimation of the higher quality projection image. The initial estimations on the variables $v_1,W_1,v_2,W_2$ are obtained by solving the following problems
\begin{equation}\label{ini:v:W:1}
  \min_{v_1,W_1} \quad\lambda^2\|v_1\|_0+\|W_1u^0-v_1\|_2^2,\quad W_1^\top W_1=I
\end{equation}
and
\begin{equation}\label{ini:v:W:2}
  \min_{v_2,W_2} \quad\lambda^2\|v_2\|_0+\|W_2f^0-v_2\|_2^2,\quad W_2^\top W_2=I
\end{equation}
using algorithm \eqref{learndict}. After the initializations, we optimize the variables $f$,$u$,$\{W_1,W_2\}$,$\{v_1,v_2\}$ in the SRD-DDTF model \eqref{obj} alternatively and iterate until convergence. Full details of the proposed algorithm is given in Algorithm \ref{alg1}. Convergence analysis of the algorithm is given in the next subsection.

\begin{algorithm}[htp]
\footnotesize
\caption{Adaptive Frames Based CT Image Reconstruction}
\label{alg1}
\textbf{Step 1.} Initialization: Compute $u^0$ from \eqref{anwave} and set $f^0=Pu^0$. Compute $v_1^0,W_1^0,v_2^0,W_2^0$ from \eqref{ini:v:W:1} and \eqref{ini:v:W:2}.\\
\textbf{Step 2.} Main Loop:\\
\While{stopping criteria are not met}
{
  \textbf{(1) optimize $f$}\begin{equation}\label{minf}
    f^{k+1}\leftarrow\text{argmin}_{f}\ \frac{\kappa}{2}\|R_\Lambda f-f_0\|_2^2+\frac{1}{2}\|R_{\Lambda^C}(Pu^k-f)\|_2^2+\frac{\mu_1}{2}\|W_1^{k} f-v_1^{k}\|_2^2+\frac{a}{2}\|f-f^{k}\|_2^2
  \end{equation}
  \textbf{(2) optimize $u$}
 \begin{equation}\label{minu}
  u^{k+1}\leftarrow\text{argmin}_u\ \frac{1}{2}\|R_{\Lambda^C}(Pu-f^{k+1})\|_2^2+\frac{1}{2}\|R_\Lambda Pu-f_0\|_2^2+\frac{\mu_2}{2}\|W_2^{k}u-v_2^{k}\|_2^2+\frac{b}{2}\|u-u^{k}\|_2^2
  \end{equation}
    \textbf{(3) optimize $W_1,W_2$}
  \begin{equation}\label{minw1w2}
  \begin{split}
   &{W}_1^{k+1}\leftarrow\mbox{argmin}_{W_1^\top W_1=I}\ \frac{\mu_1}{2}\|W_1f^{k+1}-v_1^k\|_2^2+\frac{c_1}{2}\|W_1-W_1^k\|_2^2,\cr
   &{W}_2^{k+1}\leftarrow\mbox{argmin}_{W_2^\top W_2=I}\ \frac{\mu_2}{2}\|W_2u^{k+1}-v_2^k\|_2^2+\frac{c_2}{2}\|W_2-W_2^k\|_2^2
  \end{split}
  \end{equation}
  \textbf{(4) optimize $v_1,v_2$}
  \begin{equation}\label{minv1v2}
  \begin{split}
    &{v}_1^{k+1}\leftarrow\text{argmin}_{v_1}\ \lambda_1\|v_1\|_0+\frac{\mu_1}{2}\|W_1^{k+1}f^{k+1}-v_1\|_2^2+\frac{d_1}{2}\|v_1-v_1^k\|_2^2,\cr
    &{v}_2^{k+1}\leftarrow\text{argmin}_{v_2}\ \lambda_2\|v_2\|_0+\frac{\mu_2}{2}\|W_2^{k+1}u^{k+1}-v_2\|_2^2+\frac{d_2}{2}\|v_2-v_2^k\|_2^2
  \end{split}
  \end{equation}
  }
\end{algorithm}

Note that in step 2 of Algorithm \ref{alg1} where variables are updated alternatively, we added additional $\ell_2$ terms, $\frac{a}2\|f-f^k\|^2_2$, $\frac{b}2\|u-u^k\|^2_2$,$\frac{c_1}2\|W_1-W_1^k\|^2_2$, $\frac{c_2}2\|W_2-W_2^k\|^2_2$, $\frac{d_1}2\|v_1-v_1^k\|^2_2$,$\frac{d_2}2\|v_2-v_2^k\|^2_2$, so that we can theoretically justify the convergence of the algorithm. Numerically, however, Algorithm \ref{alg1} still converges with $a=b=c_1=c_2=d_1=d_2=0$.

Problem \eqref{minf} in Algorithm \ref{alg1} has the following closed-form solution:
\begin{equation*}
f^{k+1}=(R_{\Lambda^c}+\kappa R_\Lambda+(\mu_1+a)I)^{-1}(R_{\Lambda^c}Pu^k+\kappa R_\Lambda f_0+\mu_1{W_1^{k}}^Tv_1^k+af^k),
\end{equation*}
where $R_{\Lambda^c}+\mu_1R_\Lambda+(\mu_3+a)I$ is simply a diagonal matrix and hence no matrix inversion is needed. Problem \eqref{minu} also has a closed-form solution:
\begin{equation*}
u^{k+1}=(P^TP+(\mu_2+b)I)^{-1}(P^TR_{\Lambda^c}f^{k+1}+P^TR_\Lambda f_0+\mu_2{W_2^k}^\top{v_2}^k+bu^{k}),
\end{equation*}
which can be efficiently solved by the conjugate gradient method.

The updates on the variables $v_1,W_1,v_2,W_2$ can be implemented by reformulating the problem in the form \eqref{reformulate} and solving them by a variant algorithm of \eqref{learndict}\cite{bao2014convergence}. To be more specific, we first make the following reformulations:
\begin{equation}\label{}
  \begin{aligned}
  &\{f,v_1,W_1\}\quad \Leftrightarrow\quad\{F,V_1,D_1\},\\
  &\{u,v_2,W_2\}\quad \Leftrightarrow\quad\{U,V_2,D_2\}.\\
  \end{aligned}
\end{equation}
Thus, to solve problem \eqref{minw1w2}, we can simply compute
\begin{equation}\label{}
\left\{
\begin{aligned}
  &D_1^{k+1}=X_1Y_1^\top,\quad\text{where }X_1\Sigma_1 Y_1^\top=F^{k+1}(V_1^k)^\top+\frac{c_1}{\mu_1}D_1^k;\\
  &D_2^{k+1}=X_2Y_2^\top,\quad\text{where }X_2\Sigma_2 Y_2^\top=U^{k+1}(V_2^k)^\top+\frac{c_2}{\mu_2}D_2^k.
\end{aligned}\right.
\end{equation}
To solve problem \eqref{minv1v2}, we can simply compute
\begin{equation}\label{}
\begin{aligned}
  & V_1^{k+1}=\mathcal{T}_{\sqrt{2\lambda_1/(\mu_1+d_1)}}((\mu_1(D_1^{k+1})^TF^{k+1})+d_1V_1^k)/(\mu_1+d_1)),\\
  &V_2^{k+1}=\mathcal{T}_{\sqrt{2\lambda_2/(\mu_2+d_2)}}((\mu_2(D_2^{k+1})^TU^{k+1})+d_2V_2^k)/(\mu_2+d_2)),
\end{aligned}
\end{equation}
where $\mathcal{T}_a(\cdot)$ is the hard-thresholding operator defined by \eqref{D:HT}.

\subsection{Convergence Analysis}

In this subsection, we prove that under the bounded assumption, $\{f^k,u^k,W_1^k,W_2^k, v_1^k, v_2^k\}$ generated by Algorithm \ref{alg1} converges globally and the limit is a stationary point of the proposed model \eqref{obj}. Our convergence analysis has a similar structure as that of \cite{bao2014convergence}. Our analysis is also based on the recent work by \cite{Attouch2010Proximal,bolte2014proximal,xu2014globally}, where convergence of alternative optimization algorithms on nonconvex and nonsmooth functions is studied using Kurdyka-{\L}ojasiewicz (KL) property. However, our algorithm has four block coordinates ($f,u,\{W_1,W_2\},\{v_1,v_1\}$), which leads to a relatively more complicated convergence analysis than that of \cite{bao2014convergence,Attouch2010Proximal} where algorithm of two blocks was analyzed.

Our convergence analysis is based on the following assumption:
\begin{assu}\label{ass}
The sequence $\{u^k,f^k\}$ generated by algorithm \ref{alg1} is bounded.
\end{assu}

Under this assumption, we will prove:
\begin{enumerate}
  \item Global convergence of $\{f^k,u^k,W_1^k,W_2^k, v_1^k, v_2^k\}$ using KL property;
  \item The limit is a stationary point of the SRD-DDTF model \eqref{obj}.
\end{enumerate}

We start with some basic notation and definitions.

\begin{defn}\textsc{(Critical Point)}
Let $f: \mathbb{R}^n\mapsto \mathbb{R}\cup\{\pm\infty\}$ be a proper and lower semi-continuous function.
\begin{enumerate}
  \item The Fr\'{e}chet subdifferential of $f$ at $x$ is defined by
      \begin{equation}\label{}
        \partial_Ff(x):=\{w\in \mathbb{R}^n:\liminf_{y\rightarrow x}\frac{f(y)-f(x)-\langle w,y-x\rangle}{\|y-x\|}\geqslant 0\}
       \end{equation}
       for any $x$ with $|f(x)|<\infty$ and $\partial_Ff(x)=\emptyset$ if $|f(x)|=\infty$.  Denote $\partial_F f=\{x:\partial_F f(x)\neq\emptyset\}$.
  \item The limiting-subdifferential(or simply subdifferential) of $f$ at $x$ is defined by
      \begin{equation}\label{}
        \partial f(x):=\{u\in \mathbb{R}^n: \exists x_n\rightarrow x,f(x_n)\rightarrow f(x), u_n\in\partial_F f(x_n)\rightarrow u,n\rightarrow\infty\}.
      \end{equation}
      Denote $\partial f=\{x:\partial f(x)\neq\emptyset\}$.
  \item For each $x\in \text{dom}f$, $x$ is called the \textbf{stationary point} of f if it satisfies $0\in\partial_Ff(x)$.
\end{enumerate}
\end{defn}

\begin{rk}
Our definition of stationary point is the same with the one used in \cite{bao2014convergence}, which is stronger than the definition used by \cite{Attouch2010Proximal,bolte2014proximal}.
\end{rk}

\begin{defn}\textsc{(Kurdyka-{\L}ojasiewicz Property)}\cite{kurdyka1998,lojasiewicz1993geometrie}.
The function $f$ is said to have the Kurdyka-{\L}ojasiewicz Property at $x\in\partial f $ if there exist $\eta\in(0,+\infty]$, a neighborhood $U$ of $x$ and a continuous concave function $\varphi:[0,\eta)\rightarrow [0,+\infty]$, such that:
\begin{itemize}
  \item $\varphi(0)=0,$
  \item $\varphi\in C^1\left((0,\eta)\right),$
  \item $\varphi'(z)>0,\forall z\in(0,\eta),$
  \item and $\forall x'\in U\cap \{x':f(x)<f(x')<f(x)+\eta\},$ the Kurdyka-{\L}ojasiewicz inequality holds:
      \begin{equation}\label{KL}
        \varphi'(f(x')-f(x))\text{dist}(0,\partial f(x'))\geq1.
      \end{equation}
\end{itemize}
If for all $x\in \partial f$ KL property holds, then $f$ is called a KL function.
\end{defn}

Suppose the patches of $f$ have size $n\times n$, and the patches of $u$ have size $m\times m$. For simplicity, define $\mathcal{D}_1=\{W\in R^{n^2\times n^2}:W^\top W=I_{n^2}\}$ and $\mathcal{D}_2=\{W\in R^{m^2\times m^2}:W^\top W=I_{m^2}\}$. Define
\begin{equation*}\label{}
\begin{aligned}
  Q(f,u,W_1,W_2, v_1, v_2)=&\frac{1}{2}\|R_{\Lambda^C}(Pu-f)\|_2^2+\frac{1}{2}\|R_\Lambda Pu-f_0\|_2^2+\frac{\kappa}{2}\|R_\Lambda f-f_0\|_2^2\\
  &\hspace*{2.0in}+\frac{\mu_1}{2}\|W_1 f-v_1\|_2^2+\frac{\mu_2}{2}\|W_2u-v_2\|_2^2,
\end{aligned}
\end{equation*}
and
\begin{equation*}
f_1(v_1)=\lambda_1\|v_1\|_0,\quad f_2(v_2)=\lambda_2\|v_2\|_0,\quad g_1(W_1)=I_{\mathcal{D}_1}(W_1),\quad g_2(W_2)=I_{\mathcal{D}_2}(W_2),
\end{equation*}
where $I_{\mathcal{D}}(W)=0$, if $W\in\mathcal{D}$ and $+\infty$ otherwise.

Then, problem \eqref{obj} can be reformulated as
\begin{equation}\label{obj2}
  \min_{f,u,W_1,W_2, v_1, v_2} F(f,u,W_1,W_2, v_1, v_2):=f_1(v_1)+f_2(v_2)+g_1(W_1)+g_2(W_2)+Q(f,u,W_1,W_2, v_1, v_2).
\end{equation}
It is clear that $f_1,f_2$ are lower semi-continuous. Since $\mathcal{D}_1,\mathcal{D}_2$ are compact, $g_1,g_2$ are also lower semi-continuous. For convenience, let $X^k=(f^k,u^k,W_1^k,W_2^k, v_1^k, v_2^k), Q^k=Q(X^k),F^k=F(X^k).$

\begin{lem}\label{lem1}
Denote $l=\min\{a,b,c_1,c_2,d_1,d_2\}$. We have
\begin{equation*}
\frac{l}{2}\|X^{k+1}-X^k\|_2^2\leqslant F^{k}-F^{k+1} \quad\mbox{and}\quad \|X^{k+1}-X^{k}\|_2\rightarrow 0, k\rightarrow \infty,
\end{equation*}
and the sequence $\{F^k\}$ is bounded and monotonically decreases to a limit point.
\end{lem}

\begin{proof}
According to Algorithm \ref{alg1}, we have
\begin{equation}\label{E:lem4:1}
\begin{aligned}
  &F(f^{k+1},u^{k+1},W_1^{k+1},W_2^{k+1}, v_1^{k+1}, v_2^{k+1})\\
  \leq \quad &F(f^{k+1},u^{k+1},W_1^{k+1},W_2^{k+1}, v_1^{k+1}, v_2^{k+1})+\frac{a}{2}\|f^{k+1}-f^k\|_2^2+\frac{b}{2}\|u^{k+1}-u^k\|_2^2\\
  &+\frac{c_1}{2}\|W_1^{k+1}-W_1^k\|_2^2+\frac{c_2}{2}\|W_2^{k+1}-W_2^k\|_2^2+\frac{d_1}{2}\|v_1^{k+1}-v_1^k\|_2^2+\frac{d_2}{2}\|v_2^{k+1}-v_2^k\|_2^2\\
  \leq \quad&F(f^{k},u^{k},W_1^{k},W_2^{k}, v_1^{k}, v_2^{k})\\
  \leq\quad&\dots\leq F(f^{0},u^{0},W_1^{0},W_2^{0},v_1^{0},v_2^{0}).\\
\end{aligned}
\end{equation}
Therefore, the sequence $\{F^k\geq0\}$ is bounded and monotonically decreasing, thus convergent. The second inequality of \eqref{E:lem4:1} gives us
\begin{align*}
&\frac{a}{2}\|f^{k+1}-f^k\|_2^2+\frac{b}{2}\|u^{k+1}-u^k\|_2^2+\frac{c_1}{2}\|W_1^{k+1}-W_1^k\|_2^2+\frac{c_2}{2}\|W_2^{k+1}-W_2^k\|_2^2\\
&+\frac{d_1}{2}\|v_1^{k+1}-v_1^k\|_2^2+\frac{d_2}{2}\|v_2^{k+1}-v_2^k\|_2^2
\quad \leq \quad F^k-F^{k+1},\\
\Rightarrow \quad&\frac{l}{2}\|X^{k+1}-X^k\|_2^2\quad \leq \quad F^k-F^{k+1},
\end{align*}
which leads to
$$\frac{l}{2}\sum_{i=0}^k\|X^{i+1}-X^i\|_2^2\leq F^0-F^{k+1}\leq F^0, \forall k.$$
Thus, $\sum_{k=0}^\infty\|X^{k+1}-X^k\|_2^2$ are finite and $ \|X^{k+1}-X^{k}\|_2\rightarrow 0, k\rightarrow \infty.$
\end{proof}

\begin{lem}\label{lem2}
Based on Assumption \ref{ass}, the sequence $X^k=\{f^k,u^k,W_1^k,W_2^k, v_1^k, v_2^k\}$ generated by algorithm \ref{alg1} is bounded and thus has a convergent subsequence.
\end{lem}

\begin{proof}
By Assumption \ref{ass}, the sequence $\{u^k,f^k\}$ is bounded. Also, $W_1^{k}\in\mathcal{D}_1,W_2^{k}\in\mathcal{D}_2$ are bounded too. Combining with Lemma \ref{lem1}, we have
\begin{equation*}
  \|v_1^k\|_2-\|W_1^kf^k\|_2\leq\|W_1^kf-v_1^k\|_2\leq\sqrt{\frac{2F^k}{\mu_1}}\leq\sqrt{\frac{2F^0}{\mu_1}},\quad\forall k,
\end{equation*}
hence $v_1^k$ is also bounded. Similarly, $v_2^k$ is bounded.  Thus, $X^k$ is bounded and has convergent subsequence.
\end{proof}

\begin{lem}\label{lem3}
Define
\begin{equation}\label{Ak}
\begin{aligned}
A^{k}=&\big(-a\left(f^{k}-f^{k-1}\right),-b\left(u^{k}-u^{k-1}\right),-c_1(W_1^{k}-W_1^{k-1}),-c_2(W_2^{k}-W_2^{k-1}),\\
&\hspace*{3.0in} -d_1(v_1^{k}-v_1^{k-1}),-d_2(v_2^{k}-v_2^{k-1})\big)\\
&+(Q_f(f^{k},u^{k},W_1^{k},W_2^{k},v_1^{k},v_2^{k})-Q_f(f^{k},u^{k-1},W_1^{k-1},W_2^{k-1},v_1^{k-1},v_2^{k-1}),0,0,0,0,0)\\
&+(0,Q_u(f^{k},u^{k},W_1^{k},W_2^{k},v_1^{k},v_2^{k})-Q_u(f^{k},u^{k},W_1^{k-1},W_2^{k-1},v_1^{k-1},v_2^{k-1}),0,0,0,0)\\
&+(0,0,Q_{W_1}(f^{k},u^{k},W_1^{k},W_2^{k},v_1^{k},v_2^{k})-Q_{W_1}(f^{k},u^{k},W_1^{k},W_2^{k},v_1^{k-1},v_2^{k-1}),0,0,0)\\
&+(0,0,0,Q_{W_2}(f^{k},u^{k},W_1^{k},W_2^{k},v_1^{k},v_2^{k})-Q_{W_2}(f^{k},u^{k},W_1^{k},W_2^{k},v_1^{k-1},v_2^{k-1}),0,0).\\
\end{aligned}
\end{equation}
Then $A_k\in \partial F^k$ and there exist a positive constant $M$ such that $\|A^k\|_2\leqslant M\|X^k-X^{k-1}\|_2$.
\end{lem}

\begin{proof}
From the Algorithm \ref{alg1}, we have
\begin{align*}
&0\in a(f^{k}-f^{k-1})+Q_f(f^{k},u^{k-1},W_1^{k-1},W_2^{k-1},v_1^{k-1},v_2^{k-1});\\
&0\in b(u^{k}-u^{k-1})+Q_u(f^{k},u^{k},W_1^{k-1},W_2^{k-1},v_1^{k-1},v_2^{k-1});\\
&0\in c_1(W_1^{k}-W_1^{k-1})+\partial g_1(W_1^k)+Q_{W_1}(f^{k},u^{k},W_1^{k},W_2^{k},v_1^{k-1},v_2^{k-1});\\
&0\in c_2(W_2^{k}-W_2^{k-1})+\partial g_2(W_2^k)+Q_{W_2}(f^{k},u^{k},W_1^{k},W_2^{k},v_1^{k-1},v_2^{k-1});\\
& 0\in d_1(v_1^k-v_1^{k-1})+\partial f_1(v_1^k)+Q_{v_1}(f^{k},u^{k},W_1^{k},W_2^{k},v_1^{k},v_2^{k});\\
& 0\in d_2(v_2^k-v_2^{k-1})+\partial f_2(v_2^k)+Q_{v_2}(f^{k},u^{k},W_1^{k},W_2^{k},v_1^{k},v_2^{k}).
\end{align*}
Together with \eqref{Ak}, it is easy to see that $A_k\in\partial F^k$. Note that under the bounded setting, $\partial Q$ is Lipschitz continuous. Letting its Lipschitz constant be $L$, we have $\|A^k\|_2\leqslant M\|X^k-X^{k-1}\|_2$ where $M=\sqrt{6}\max\{a,b,c_1,c_2,d_1,d_2\}+4L.$
\end{proof}

It has been proved in \cite{bao2014convergence} that $f_1,f_2,g_1,g_2$ are KL functions and $Q$ is a polynomial function, hence also a KL function. Therefore, our object function $F$ is a KL function. Then, we have the following theorem.

\begin{thm}\label{thm1}
Under the Assumption \ref{ass}, denote a point $\bar{X}=\{\bar{f},\bar{u},\bar{W}_1,\bar{W}_2,\bar{v}_1,\bar{v}_2\}$ and its  neighborhood $U$, $\eta$ and  continuous concave function $\varphi$ for the consideration of KL property of $F$ at $\bar{X}$, that is
$$
\varphi'(F(X)-F(\bar{X}))dist(0,\partial F(X))\geq 1,\quad \forall X\in U\bigcap\{X| F(\bar{X})< F(X)<F(\bar{X})+\eta\}.
$$
Denote a sequence $\{X^k\}$ generated by Algorithm \ref{alg1} starting from $X^0$ and a constant $r>0$ with $B(\bar{X},r)\subset U$, we assume that
 \begin{equation*}
\bar{F}<F^k<\bar{F}+\eta,
\end{equation*}
and
\begin{equation}\label{X0}
2\sqrt{\frac{2}{l}(F^0-\bar{F})}+\frac{2M}{l}\phi(F^0-\bar{F})+\|X^0-\bar{X}\|_2<r
\end{equation}
where $l$ is the constant in Lemma \ref{lem1} and $M$ is the constant in Lemma \ref{lem3}. Then, we conclude that
\begin{enumerate}
\item $X^k\in B(\bar{X},r),\quad k\geq1$;
\item $\sum_{k=1}^\infty \|X^{k+1}-X^{k}\|_2<\infty$, which means $\{X^k\}$ is a Cauchy sequence and converges globally.
\end{enumerate}
\end{thm}

\begin{proof}
For simplicity, we use $\|\cdot\|$ to denote $\|\cdot\|_2$, and without loss of generality, we assume $\bar{F}=F(\bar{X})=0$.

First, it is clear from condition \eqref{X0} that $X^0\in B(\bar{X},r)$. Then, we consider the case $k=1$. Combined with Lemma \ref{lem1}, we have
\begin{equation}\label{}
  \|X^1-\bar{X}\|\leq\|X^1-X^0\|+\|X^0-\bar{X}\|\leq\sqrt{\frac{2}{l}F^0}+\|X^0-\bar{X}\|.
\end{equation}
Hence, by condition \eqref{X0}, we have $X^1\in B(\bar{X},r)$. Now, we will use induction to prove the two conclusions.

Supposed $X^i\in B(\bar{X},r), 1\leq i\leq k$, using KL property at $\bar{X}$ and Lemma \ref{lem3}, we have
$$
\varphi'(F^i)\geq \frac{1}{\|\partial F^i\|}\quad\Rightarrow\quad\varphi'(F^i)\geq \frac{1}{\|A_i\|}\geq\frac{1}{M\|X^i-X^{i-1}\|}.
$$
Combined with lemma \ref{lem1}, we have
$$
\varphi'(F^i)(F^{i}-F^{i+1})\geq\frac{l\|X^{i+1}-X^i\|^2}{2M\|X^i-X^{i-1}\|}.
$$

Notice that $\varphi$ is concave, so we have
$$
\varphi(F^{i})-\varphi(F^{i+1})\geq\frac{l\|X^{i+1}-X^i\|^2}{2M\|X^i-X^{i-1}\|}.
$$
Thus
\begin{align*}
&\frac{2M}{l}\|X^i-X^{i-1}\|(\varphi(F^{i})-\varphi(F^{i+1}))\geq \|X^{i+1}-X^i\|^2\\
\Rightarrow\quad & \sqrt{\|X^i-X^{i-1}\|\cdot\frac{2M}{l}(\varphi(F^{i})-\varphi(F^{i+1}))}\geq \|X^{i+1}-X^i\|\\
\Rightarrow\quad & \|X^i-X^{i-1}\|+\frac{2M}{l}(\varphi(F^{i})-\varphi(F^{i+1}))\geq 2\|X^{i+1}-X^i\|.
\end{align*}

Sum it up for $1\leq i\leq k$, then we have
\begin{equation}\label{sum}
\begin{split}
&\|X^1-X^0\|+\frac{2M}{l}(\varphi(F^1)-\varphi(F^{k+1}))\geq \sum_{i=1}^k\|X^{i+1}-X^i\|+\|X^{k+1}-X^k\|\\
\Rightarrow\quad &\|X^1-X^0\|+\frac{2M}{l}\varphi(F^1)\geq\sum_{i=1}^k\|X^{i+1}-X^i\|.
\end{split}
\end{equation}
Therefore,
\begin{align*}
&\|X^{k+1}-\bar{X}\|\leq \sum_{i=1}^k\|X^{i+1}-X^i\|+\|X^1-\bar{X}\|\\
\Rightarrow\quad& \|X^{k+1}-\bar{X}\|\leq \|X^1-X^0\|+\frac{2M}{l}\varphi(F^1)+\|X^1-\bar{X}\|\\
\Rightarrow\quad& \|X^{k+1}-\bar{X}\|\leq 2\sqrt{\frac{2}{l}F^0}+\frac{2M}{l}\varphi(F^0)+\|X^0-\bar{X}\|.
\end{align*}
From condition \eqref{X0}, we conclude that $X^{k+1}\in B(\bar{X},r)$. So we have proven the first conclusion. The second conclusion is straightforward from inequality \eqref{sum}.
\end{proof}

\begin{lem}\label{lem4}
Under the Assumption \ref{ass}, for any convergent subsequence $X^{k'}$  with limit point $X^*=(f^*,u^*,W_1^*,W_2^*,v_1^*,v_2^*)$, we have
\begin{equation*}
  X^{k'-1}\rightarrow X^*, k'\rightarrow\infty;\quad
\end{equation*}
and
\begin{equation*}
  \lim_{k'\rightarrow\infty} f_1(v_1^{k'})+f_2(v_2^{k'})=f_1(v_1^{*})+f_2(v_2^{*}), \lim_{k'\rightarrow\infty} F(X^{k'})=F(X^*).
\end{equation*}
\end{lem}

\begin{proof}
From Lemma \ref{lem1}, we have $\lim_{k'\rightarrow\infty}\|X^{k'-1}-X^{k'}\|_2=0$. Since $\lim_{k'\rightarrow\infty} X^{k'}=X^*$, for any $\epsilon>0$, there exists $K$, when $k'>K$,
\begin{align*}
&\|X^{k'-1}-X^{k'}\|_2<\epsilon/2, \quad \|X^{k'}-X^{*}\|_2<\epsilon/2\\
\Rightarrow\quad&\|X^{k'-1}-X^*\|_2\leq\|X^{k'-1}-X^{k'}\|_2+\|X^{k'}-X^{*}\|_2<\epsilon.
\end{align*}
Therefore, $\lim_{k'\rightarrow\infty}X^{k'-1}=X^*$.

From \eqref{minv1v2} in Algorithm \ref{alg1}, we have
\begin{equation*}\label{}
  \begin{aligned}
  &Q(u^{k'},f^{k'},v_1^{k'},W_1^{k'},v_2^{k'},W_2^{k'})+f_1(v_1^{k'})+f_2(v_2^{k'})+\frac{d_1}{2}\|v_1^{k'}-v_1^{k'-1}\|_2^2+\frac{d_2}{2}\|v_2^{k'}-v_2^{k'-1}\|_2^2\\
  \le &Q(u^{k'},f^{k'},v_1,W_1^{k'},v_2,W_2^{k'})+f_1(v_1)+f_2(v_2)+\frac{d_1}{2}\|v_1-v_1^{k'-1}\|_2^2+\frac{d_2}{2}\|v_2-v_2^{k'-1}\|_2^2,\quad\forall v_1,v_2.
  \end{aligned}
\end{equation*}
Replacing $v_1,v_2$ with $v_1^*,v_2^*$, and taking $k'$ to infinity, we have
\begin{equation}\label{lim}
\liminf _{k'\rightarrow\infty}f_1(v_1^{k'})+f_2(v_2^{k'})\leq f_1(v_1^*)+f_2(v_2^*).
\end{equation}
Note that $f_1$ and $f_2$ are lower semi-continuous. Together with \eqref{lim}, we have
\begin{equation}\label{E:lem2}
   \liminf_{k'\rightarrow+\infty}f_1(v_1^{k'})+f_2(v_2^{k'})= f_1(v_1^*)+f_2(v_2^*).
\end{equation}
On the other hand, we have $W_1^{k'}\in\mathcal{D}_1,W_2^{k'}\in\mathcal{D}_2$. Since $\mathcal{D}_1,\mathcal{D}_2$ are compact, we have $g_1(W_1^*)=g_1(W_1^{k'})=g_2(W_2^*)=g_2(W_2^{k'})=0$ for all $k'$. Since $Q$ is continuous, we have $$\lim_{k'\rightarrow +\infty}Q(u^{k'},f^{k'},v_1^{k'},W_1^{k'},v_2^{k'},W_2^{k'})=Q(u^*,f^*,v_1^*,W_1^*,v_2^*,W_2^*).$$ Note that $\{F^{k'}\}$ is monotonically decreasing and bounded. Thus, $F^{k'}$ is convergent which means that the liminf in \eqref{E:lem2} is in fact a regular limit. Consequently, we have $\lim_{k'\rightarrow\infty} F(X^{k'})=F(X^*)$.
\end{proof}

\begin{thm}\label{thm2}
\textsc{(Global Convergence)} Under the Assumption \ref{ass}, the sequence $\{X^k\}$ generated by Algorithm \ref{alg1} is globally convergent.
\end{thm}

\begin{proof}
From Lemma \ref{lem2}, we know the sequence$\{X^k\}$ generated by algorithm\ref{alg1} has a convergent subsequence $\{X^{k'}\}$ with $\lim_{k'\rightarrow\infty}X^{k'}=X^*$. Let $F^k=F(X^k)$ and $F^*=F(X^*)$. Lemma \ref{lem4} tells us that $\lim_{k'\to\infty} F^{k'}=F^*$.

Suppose there is a $k$ such that $F^k=F^*$. Then, Lemma \ref{lem1} implies that $F^j=F^*$ for all $j\ge k$, and hence $X^j=X^*$ for all $j\ge k$. Therefore, $X^k$ converges.

Suppose $F^k>F^*$ for all $k$. Take $\bar{X}$ to be $X^*$ in Theorem \ref{thm1}. Choose $K'\in\{k'\}$ large enough so that $2\sqrt{\frac{2}{l}(F^{K'}-\bar{F})}+\frac{2M}{l}\phi(F^{K'}-\bar{F})+\|X^{K'}-\bar{X}\|<r$. Then, we have global convergence by taking $X^{K'}$ as the new initial point $X^0$ in Theorem \ref{thm1}.
\end{proof}

In the following, we will prove that the limit of $\{X^k\}$ is a stationary point of our SRD-DDTF model.

\begin{thm}\textsc{(Stationary Point)}
Under Assumption \ref{ass}, the sequence $X^k:=(f^k,u^k,W_1^k,W_2^k,v_1^k,v_2^k)$ globally converges to a stationary point of the SRD-DDTF model \eqref{obj}.
\end{thm}
\begin{proof}
Theorem \ref{thm2} tells us that $X^k$ is globally convergent. Denote its limit point by $X^*=(f^*,u^*,W_1^*,W_2^*,v_1^*,v_2^*)$ and $F(X^*)=F^*$. From Algorithm \ref{alg1}, we have the following four inequalities:
\begin{equation*}\label{}
  \begin{aligned}
  &Q(f^{k+1},u^k,W_1^k,W_2^k,v_1^k,v_2^k)+\frac{a}{2}\|f^{k+1}-f^k\|_2^2\\
  \leq\quad& Q(f,u^k,W_1^k,W_2^k,v_1^k,v_2^k)+\frac{a}{2}\|f-f^k\|_2^2,\quad \forall f;\\
  &Q(f^{k+1},u^{k+1},W_1^k,W_2^k,v_1^k,v_2^k)+\frac{b}{2}\|u^{k+1}-u^k\|_2^2\\
  \leq\quad& Q(f^{k+1},u,W_1^k,W_2^k,v_1^k,v_2^k)+\frac{b}{2}\|u-u^k\|_2^2, \quad \forall u;\\
  &Q(f^{k+1},u^{k+1},W_1^{k+1},W_2^{k+1},v_1^k,v_2^k)+g_1(W_1^{k+1})+g_2(W_2^{k+1})+\frac{c_1}{2}\|W_1^{k+1}-W_1^k\|_2^2+\frac{c_2}{2}\|W_2^{k+1}-W_2^k\|_2^2\\
  \leq\quad & Q(f^{k+1},u^{k+1},W_1,W_2,v_1^k,v_2^k)+g_1(W_1)+g_2(W_2)+\frac{c_1}{2}\|W_1-W_1^k\|_2^2+\frac{c_2}{2}\|W_2-W_2^k\|_2^2, \quad\forall W_1,W_2;\\
  &Q(f^{k+1},u^{k+1},W_1^{k+1},W_2^{k+1},v_1^{k+1},v_2^{k+1})+f_1(v_1^{k+1})+f_2(v_2^{k+1})+\frac{d_1}{2}\|v_1^{k+1}-v_1^k\|_2^2+\frac{d_2}{2}\|v_2^{k+1}-v_2^k\|_2^2\\
  \leq\quad & Q(f^{k+1},u^{k+1},W_1^{k+1},W_2^{k+1},v_1,v_2)+f_2(v_2)+f_1(v_1)+\frac{d_1}{2}\|v_1-v_1^k\|_2^2+\frac{d_2}{2}\|v_2-v_2^k\|_2^2, \quad\forall v_1,v_2.\\
  \end{aligned}
\end{equation*}
Taking $k\rightarrow\infty$, we have
\begin{equation}\label{coordinate}
  \left\{\begin{aligned}
  &F^*\leq F(f^*+\delta f,u^*,W_1^*,W_2^*,v_1^*,v_2^*)+\frac{a}{2}\|\delta f\|_2^2,\quad\forall \delta f;\\
  &F^*\leq F(f^*,u^*+\delta u,W_1^*,W_2^*,v_1^*,v_2^*)+\frac{b}{2}\|\delta u\|_2^2,\quad\forall \delta u;\\
    &F^*\leq F(f^*,u^*,W_1^*+\delta W_1,W_2^*+\delta W_2,v_1^*,v_2^*)+\frac{c_1}{2}\|\delta W_1\|_2^2+\frac{c_2}{2}\|\delta W_2\|_2^2,\quad\forall \delta W_1,\delta W_2;\\
  &F^*\leq F(f^*,u^*,W_1^*,W_2^*,v_1^*+\delta v_1,v_2^*+\delta v_2)+\frac{d_1}{2}\|\delta v_1\|_2^2+\frac{d_2}{2}\|\delta v_2\|_2^2,\quad\forall \delta v_1,\delta v_2.\\
  \end{aligned}\right.
\end{equation}

Therefore, for any $\delta X=(\delta f,\delta u,\delta W_1, \delta W_2,\delta v_1,\delta v_2)$, we have
\begin{equation}\label{}
  \begin{aligned}
  &\quad\liminf_{\|\delta X\|\rightarrow 0} \frac{F(X^*+\delta X)-F(X^*)}{\|\delta X\|}\\
  &= \liminf_{\|\delta X\|\rightarrow 0}  \frac{Q(X^*+\delta X)-Q(X^*)+(f_1(v_1^*+\delta v_1)+f_2(v_2^*+\delta v_2))-(f_1(v_1^*)+f_2(v_2^*))}{\|\delta X\|}\\
  \quad &+\frac{(g_1(W_1^*+\delta W_1)+g_2(W_2^*+\delta W_2))-(g_1(W_1^*)+g_2(W_2^*))}{\|\delta X\|}\\
    &=\liminf_{\|\delta X\|\rightarrow 0}  \frac{\langle \nabla Q(X^*),\delta X\rangle+(f_1(v_1^*+\delta v_1)+f_2(v_2^*+\delta v_2))-(f_1(v_1^*)+f_2(v_2^*))}{\|\delta X\|}\\
   \quad & +\frac{(g_1(W_1^*+\delta W_1)+g_2(W_2^*+\delta W_2))-(g_1(W_1^*)+g_2(W_2^*))+o(\|\delta X\|)}{\|\delta X\|}\\
  &=  \liminf_{\|\delta X\|\rightarrow 0} \left(\frac{Q(f^*,u^*,W_1^*,W_2^*,v_1^*+\delta v_1,v_2^*+\delta v_2)-Q(X^*)}{\|\delta X\|}\right.\\
  \quad &\quad\quad+\frac{f_1(v_1^*+\delta v_1)+f_2(v_2^*+\delta v_2)-f_1(v_1^*)-f_2(v_2^*)+o(\|\delta v_1\|+\|\delta v_2\|)+o(\|\delta X\|)}{\|\delta X\|}\\
  \quad &+\frac{Q(f^*,u^*,W_1^*+\delta W_1,W_2^*+\delta W_2,v_1^*,v_2^*)-Q(X^*)}{\|\delta X\|}\\
  \quad &\quad\quad+\frac{g_1(W_1^*+\delta W_1)+g_2(W_2^*+\delta W_2)-g_1(W_1^*)-g_2(W_2^*)+o(\|\delta W_1\|+\|\delta W_2\|)}{\|\delta X\|}\\
  \quad &+ \frac{Q(f^*,u^*+\delta u,W_1^*,W_2^*,v_1^*,v_2^*)-Q(X^*)+o(\|\delta u\|)}{\|\delta X\|}\\
  \quad &+\left. \frac{Q(f^*+\delta f,u^*,W_1^*,W_2^*,v_1^*,v_2^*)-Q(X^*)+o(\|\delta f\|)}{\|\delta X\|}\right)\\
  &\geqslant\liminf_{\|\delta X\|\rightarrow 0} \frac{-\frac{a}2\|\delta f\|^2-\frac{b}2\|\delta u\|^2-\frac{c_1}2\|\delta W_1\|^2-\frac{c_2}2\|\delta W_2\|^2-\frac{d_1}2\|\delta v_1\|^2-\frac{d_2}2\|\delta v_2\|^2}{\|\delta X\|}\\
  &+\frac{o(\|\delta v_1\|+\|\delta v_2\|)+o(\|\delta u\|+\|\delta f\|+\|\delta W_1\|+\|\delta W_2\|)+o(\|\delta X\|)}{\|\delta X\|}\\
  &=0.
  \end{aligned}
\end{equation}
where the second and third identity follow from the first order Taylor expansion of $Q$ and the inequality follows from the inequalities \eqref{coordinate}. By definition, the limit point $X^*$ is a stationary point of our model \eqref{obj}.
\end{proof}

\section{Numerical Experiments}\label{Sec:Numerical}

It has been shown in \cite{Dong2013X} that wavelets based inpainting model \eqref{primal} has better image restoration performance than TV-based model and wavelet analysis model. Therefore, in this section, we will focus on comparing our proposed SRD-DDTF model \eqref{obj} with wavelet frame based model \eqref{primal} using the same initial value given by the analysis model \eqref{anwave}. We show that the SRD-DDTF model can achieve noticeably better image reconstruction results.

Throughout our experiments, all data is synthesized by $f=Pu+\epsilon$, where $\epsilon$ is some Gaussian white noise. The standard deviation of noise is chosen to be $\max(|f|)/300$. In our model \eqref{obj}, we always set $\kappa=1$, which is the same as the parameter $\kappa$ in \eqref{primal}. Empirically, we observe that $\mu_1\thickapprox5200\lambda_1$, $\mu_2\thickapprox8400\lambda_2$ is a good choice. The patch size for $u$ is $8 \times 8$, and the patch size for $f$ is $8\times2$ to properly adapt to the shape of matrix $f$ which has much more rows (number of detectors) than columns (number of angular projections). Parameters $\lambda_1$ and $\lambda_2$, which are the same as the parameter $\lambda_1$ and $\lambda_2$ in model \eqref{primal}, vary case by case and are chosen manually for optimal image reconstruction results. We use the analysis model \eqref{anwave} to obtain initial estimates for both our algorithm and model \eqref{primal}. We find in our experiments that the value of $\lambda$ in the analysis model has small effects on the performance of both algorithms.

The experiments are conducted, with different configurations, on a real patient's image data (provided by Dr. Xun Jia from Department of Radiation Oncology, University of Texas, Southwestern Medical Center) and the popular test data set NURBS-based cardiac-torso (NCAT) phantom \cite{Segars2000Development}. We shall refer to the former simply as ``head" and the latter as ``NCAT". In addition to visual observation given by Figure \ref{data1figure} and Figure \ref{data2figure}, we use relative error, correlation and computation time to quantify the quality of the model \eqref{primal} and our SRD-DDTF model \eqref{obj} (see Table \ref{data1}). The relative error and correlation for reconstructed $u$ based on ground truth image $u_t$ are defined as follows:
\begin{equation}\label{}
  err(u_t,u)=\frac{\|u-u_t\|_2}{\|u_t\|_2},\quad corr(u_t,u)=\frac{(u-\bar{u})(u_t-\bar{u}_t)}{\|u-\bar{u}\|_2\|u_t-\bar{u}_t\|_2}
\end{equation}
where $\bar{u}$, $\bar{u}_t$ denote the mean value of $u$ and $u_t$. As we can see from Figure \ref{data1figure}, Figure \ref{data2figure} and Table \ref{data1} that our SRD-DDTF model \eqref{obj} managed to achieve better image reconstruction results than the model \eqref{primal} of \cite{Dong2013X} for all configurations.

The stopping criterion we used takes the form $err(u^{k},u^{k-1})\leq \epsilon$ together with a maximum allowable iteration 1000, where $\epsilon=0.001$ for the image ``head" and $\varepsilon=0.005$ for ``NCAT". We consider the configurations with $N_P=15,30,45,60$ for ``head" and $N_P=60,75,90$ for ``NCAT". Table \ref{data1} shows that the SRD-DDTF model achieves noticeably better reconstruction with less relative errors and higher correlations based on ground truth images. It is remarkable that for ``NCAT", the results of the SRD-DDTF model can even achieve better image restoration than that of model \eqref{primal} using the next larger projection number. The reconstructed image for each configuration is shown in Figure \ref{data1figure} and Figure \ref{data2figure}, and it is worth noticing that the SRD-DDTF model is managed to recover some key structures that are lost by model \eqref{primal}. In particular, we list some zoom-in views of the results in Figure \ref{loc} for $N_p=15$ to show that our model is capable of restoring subtle features.

\begin{table}
  \centering
  \begin{tabular}{|c|c|c|c|c|c|c|c|c|}
    \hline
    \multicolumn{9}{|c|}{``Head"}\\
    \hline
    \multirow{2}*{NP}& \multicolumn{2}{c|}{initial value}  &  \multicolumn{3}{c|}{Model \eqref{primal} of \cite{Dong2013X}}&\multicolumn{3}{c|}{SRD-DDTF model \eqref{obj}} \\
    \cline{2-9}
     & err & corr & err & corr & time & err & corr & time \\
     \hline
    15 & 14.09 & 98.29 & 12.70 & 98.61 & 288.91 & 10.50 & 99.05 & 629.92 \\
    30 & 6.79 & 99.61 & 6.25 & 99.67 & 1186.54 & 5.39 & 99.75 & 930.65 \\
    45 & 5.20 & 99.77 &4.70& 99.81 &1550.63 &4.24& 99.85 & 1389.01\\
    60 & 4.16 & 99.85 & 3.89 & 99.87 & 319.29 & 3.58 & 99.89 & 1785.50 \\
    \hline
    \multicolumn{9}{|c|}{``NCAT"}\\
    \hline
        \multirow{2}*{NP}& \multicolumn{2}{c|}{initial value}  &  \multicolumn{3}{c|}{Model \eqref{primal} of \cite{Dong2013X}}&\multicolumn{3}{c|}{SRD-DDTF model \eqref{obj}} \\
    \cline{2-9}

     & err & corr & err & corr & time & err & corr & time \\ \hline
     60& 9.55& 99.35 & 5.00 &99.82 & 239.47&4.23 & 99.87&955.67\\
     75& 9.02& 99.42 & 4.61 &99.85 &296.71& 3.65 & 99.90&1681.94\\
     90 & 8.81 & 99.45& 4.21 & 99.87 & 303.98 & 3.41 & 99.91 & 2339.88 \\
     \hline
  \end{tabular}
  \caption{Comparison of relative errors (in percentage), correlations (in percentage) and running time (in seconds). }\label{data1}
\end{table}

\begin{figure}[htp]
  \centering
  \includegraphics[width=9cm]{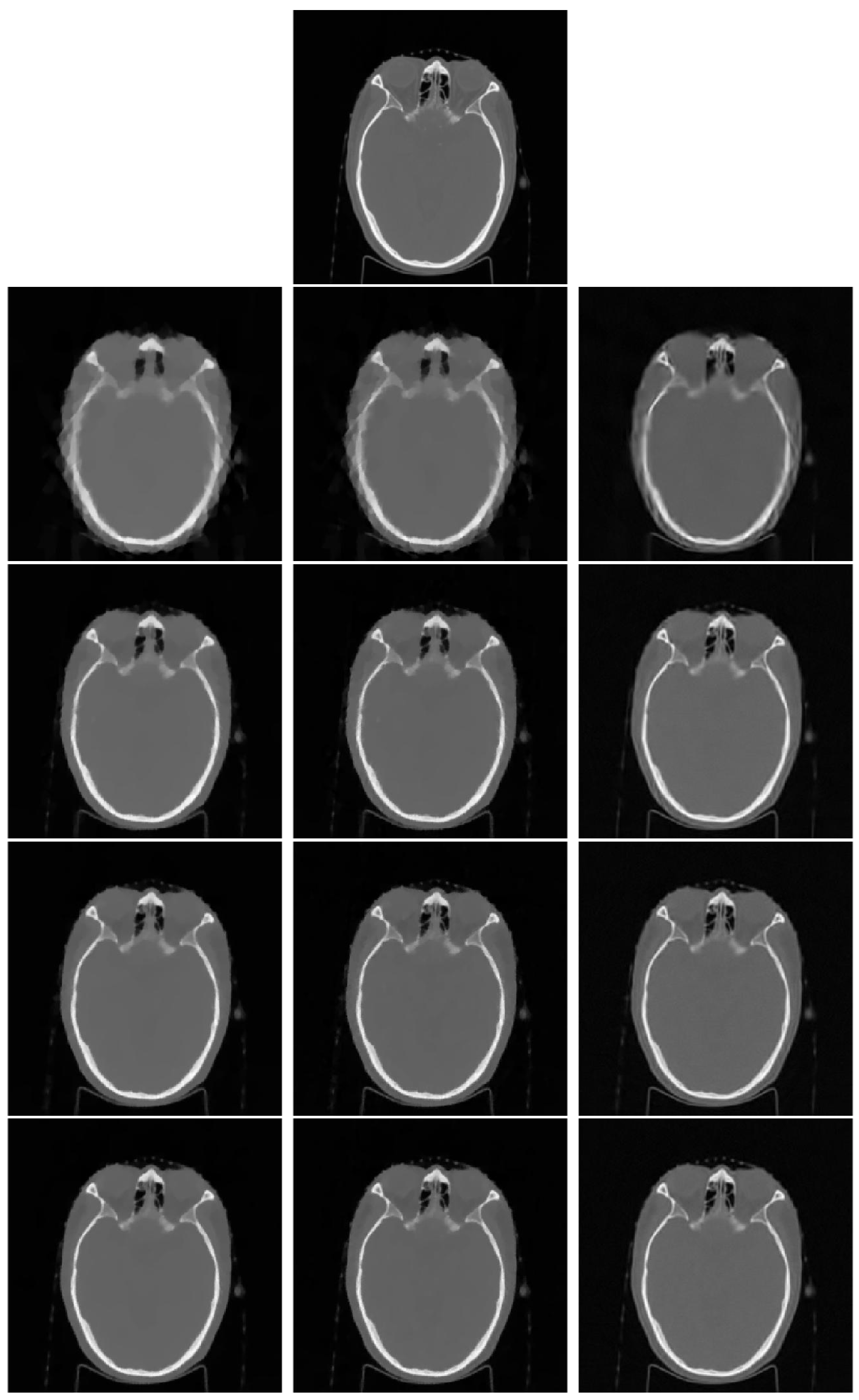}
  \caption{The tomographic results for the image ``head". The image on the top is the true data. The following rows represent the results using 15,30,45,60 projections, respectively. Images from left to right in each row are results from initial value, wavelets based inpainting model and our SRD-DDTF model.}\label{data1figure}
\end{figure}

\begin{figure}[htp]
  \centering
  \includegraphics[width=14cm]{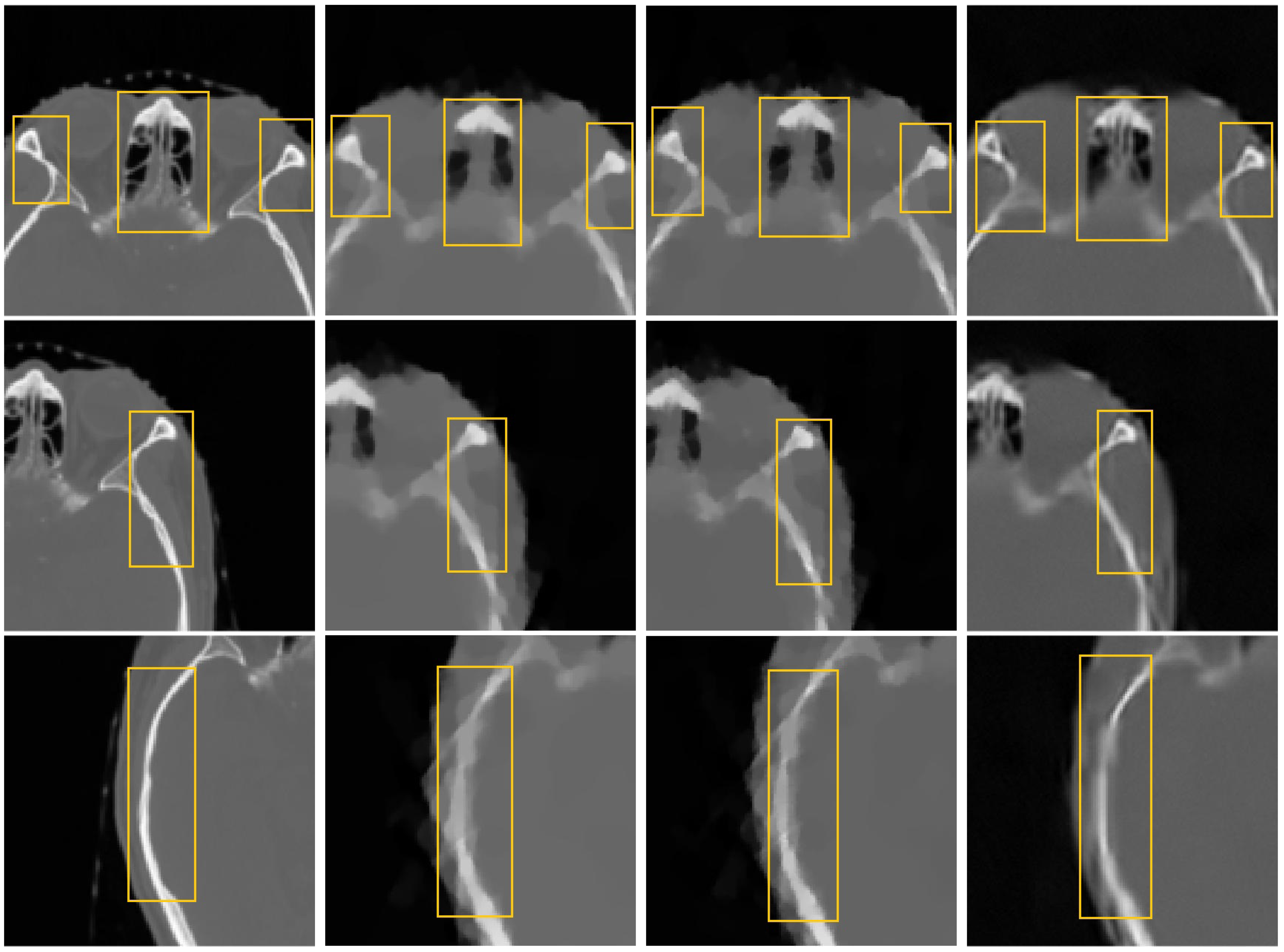}
  \caption{Local comparison of restored images for the image ``head" with $N_p=15$. Images from left to right in each row are zoom-in patterns from the ground truth, the initial value, wavelets based inpainting model restored image, and our SRD-DDTF model.}\label{loc}
\end{figure}

\begin{figure}[htp]
  \centering
  \includegraphics[width=12cm]{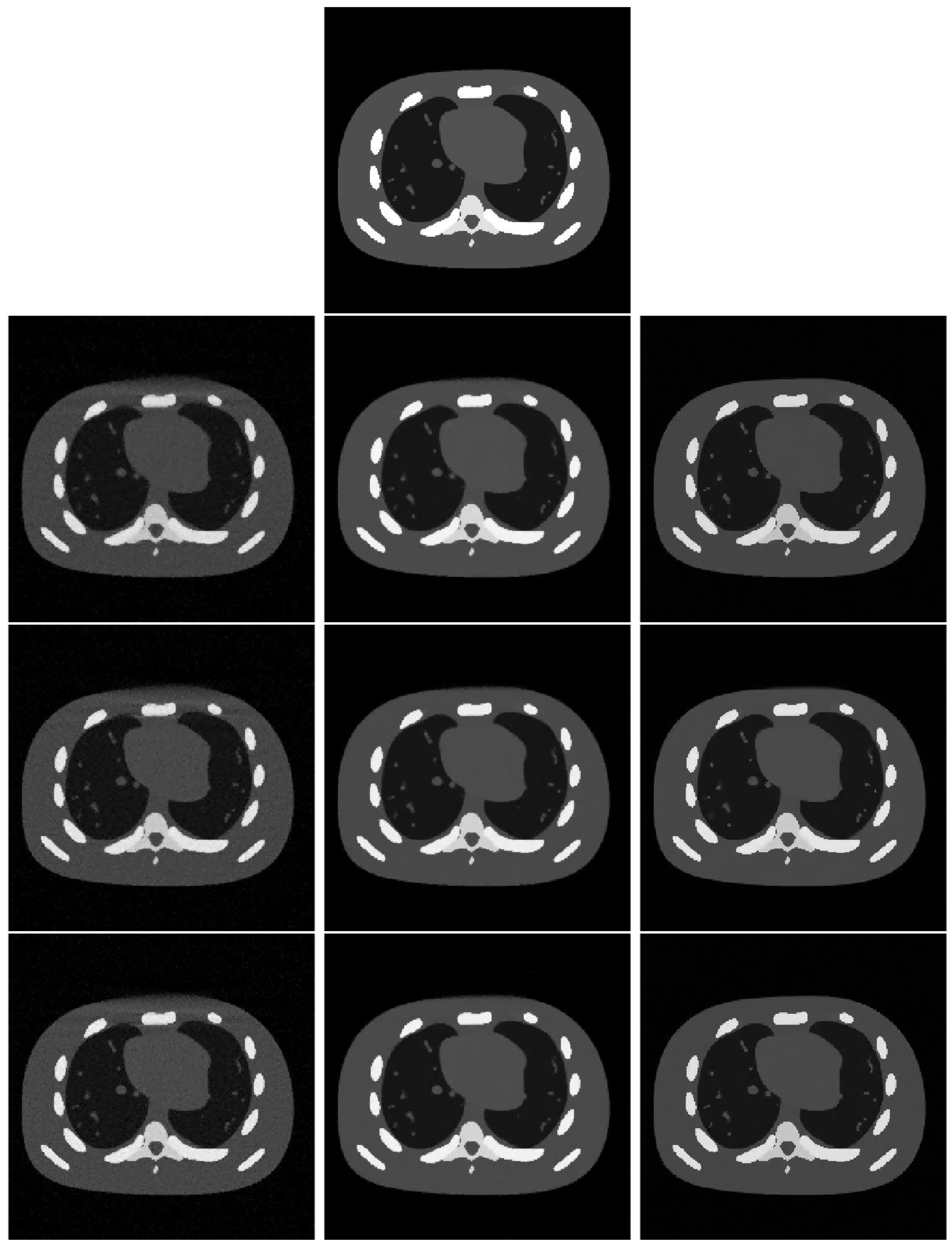}
  \caption{The tomographic results for ``NCAT". The image on the top is the true data. The following rows represent the results using 60,75,90 projections, respectively. Images from left to right in each row are results from initial value, wavelets based inpainting model and our SRD-DDTF model.}\label{data2figure}
\end{figure}

\section{Conclusion}\label{Sec:Conclusion}
In this paper, we proposed a new spatial-Radon domain CT image reconstruction model based on data-driven tight frames (SRD-DDTF), together with an efficient alternative minimization algorithm. Our convergence analysis on the proposed algorithm indicated that, under suitable assumptions, the sequence generated by the algorithm converges to a stationary point of the proposed model. Our numerical experiments showed that our model \eqref{obj} can obtain noticeably better reconstruction results than those from the model \eqref{primal}, which showed that using data-driven tight frames as sparsity priors for both the CT image and the projection image performs better than using pre-determined systems, such as tight wavelet frames, that may not be ideal for a specifically given image data.

\bibliographystyle{ieeetr}
\bibliography{ReferenceLibrary}

\begin{thebibliography}{10}

\bibitem{Dong2013X}
B.~Dong, J.~Li, and Z.~Shen, ``X-ray ct image reconstruction via wavelet frame
  based regularization and radon domain inpainting,'' {\em Journal of
  Scientific Computing}, vol.~54, no.~2-3, pp.~333--349, 2013.

\bibitem{cai2014data}
J.~F. Cai, H.~Ji, Z.~Shen, and G.~B. Ye, ``Data-driven tight frame construction
  and image denoising,'' {\em Applied \& Computational Harmonic Analysis},
  vol.~37, no.~1, p.~89¨C105, 2014.

\bibitem{Islam2006Patient}
M.~K. Islam, B.~D. Purdie~TGNorrlinger, H.~Alasti, D.~J. Moseley, M.~B. Sharpe,
  J.~H. Siewerdsen, and D.~A. Jaffray, ``Patient dose from kilovoltage cone
  beam computed tomography imaging in radiation therapy,'' {\em Medical
  Physics}, vol.~33, no.~6, pp.~1573--1582, 2006.

\bibitem{brenner2007computed}
D.~J. Brenner and E.~J. Hall, ``Computed tomography¡ªan increasing source of
  radiation exposure,'' {\em New England Journal of Medicine}, vol.~357,
  no.~22, pp.~2277--2284, 2007.

\bibitem{Ding2009Radiation}
G.~X. Ding and C.~W. Coffey, ``Radiation dose from kilovoltage cone beam
  computed tomography in an image-guided radiotherapy procedure,'' {\em
  International Journal of Radiation Oncology Biology Physics}, vol.~73, no.~2,
  p.~610¨C617, 2009.

\bibitem{Katsevich2002THEORETICALLY}
A.~Katsevich, ``Theoretically exact filtered backprojection-type inversion
  algorithm for spiral ct,'' {\em Siam Journal on Applied Mathematics},
  vol.~62, no.~6, pp.~p¨¢gs. 2012--2026, 2002.

\bibitem{Antsaklis1991Stability}
Z.~G. P.~J. Antsaklis, ``Stability of the pseudo-inverse method for
  reconfigurable control systems,'' {\em International Journal of Control},
  vol.~53, no.~3, pp.~717--729, 1991.

\bibitem{Tikhonov1978Solutions}
A.~N. Tikhonov and V.~Y. Arsenin, ``Solutions of ill-posed problems,'' {\em
  Mathematics of Computation}, vol.~32, no.~5, pp.~491--491, 1978.

\bibitem{Siddon1985Fast}
R.~Siddon, ``{Fast calculation of the exact radiological path for a
  three-dimensional CT array},'' {\em Medical Physics}, vol.~12, p.~252, 1985.

\bibitem{Dong2015}
B.~Dong and Z.~Shen, ``Image restoration: a data-driven perspective.,'' {\em
  Proceedings of the International Congress of Industrial and Applied
  Mathematics (ICIAM)}, pp.~65--108, 2015.

\bibitem{cai2008framelet}
J.~Cai, R.~Chan, and Z.~Shen, ``{A framelet-based image inpainting
  algorithm},'' {\em Applied and Computational Harmonic Analysis}, vol.~24,
  no.~2, pp.~131--149, 2008.

\bibitem{chan2003wavelet}
R.~Chan, T.~Chan, L.~Shen, and Z.~Shen, ``{Wavelet algorithms for
  high-resolution image reconstruction},'' {\em SIAM Journal on Scientific
  Computing}, vol.~24, no.~4, pp.~1408--1432, 2003.

\bibitem{cai2009split}
J.~Cai, S.~Osher, and Z.~Shen, ``{Split Bregman methods and frame based image
  restoration},'' {\em Multiscale Modeling and Simulation: A SIAM
  Interdisciplinary Journal}, vol.~8, no.~2, pp.~337--369, 2009.

\bibitem{cai2009linearized}
J.~Cai, S.~Osher, and Z.~Shen, ``{Linearized Bregman iterations for frame-based
  image deblurring},'' {\em SIAM J. Imaging Sci}, vol.~2, no.~1, pp.~226--252,
  2009.

\bibitem{zhang2011l0min}
Y.~Zhang, B.~Dong, and Z.~Lu, ``$\ell_0$ minimization of wavelet frame based
  image restoration,'' {\em Mathematics of Computation}, vol.~82,
  pp.~995--1015, 2013.

\bibitem{dong2011efficient}
B.~Dong and Y.~Zhang, ``An efficient algorithm for $\ell_0$ minimization in
  wavelet frame based image restoration,'' {\em Journal of Scientific
  Computing}, vol.~54 (2-3), pp.~350--368, 2013.

\bibitem{liang2013wavelet}
J.~Liang, J.~Li, Z.~Shen, and X.~Zhang, ``Wavelet frame based color image
  demosaicing,'' {\em Inverse Problems and Imaging}, vol.~7, no.~3,
  pp.~777--794, 2013.

\bibitem{hou2013recovering}
L.~Hou, H.~Ji, and Z.~Shen, ``Recovering over-/underexposed regions in
  photographs.,'' {\em SIAM J. Imaging Sciences}, vol.~6, no.~4,
  pp.~2213--2235, 2013.

\bibitem{cai2009blind}
J.~Cai, H.~Ji, C.~Liu, and Z.~Shen, ``{Blind motion deblurring using multiple
  images},'' {\em Journal of Computational Physics}, vol.~228, no.~14,
  pp.~5057--5071, 2009.

\bibitem{cai2009blindcvpr}
J.~Cai, H.~Ji, C.~Liu, and Z.~Shen, ``Blind motion deblurring from a single
  image using sparse approximation,'' in {\em Computer Vision and Pattern
  Recognition, 2009. CVPR 2009. IEEE Conference on}, pp.~104--111, IEEE, 2009.

\bibitem{dong2012wavelet}
B.~Dong, H.~Ji, J.~Li, Z.~Shen, and Y.~Xu, ``Wavelet frame based blind image
  inpainting,'' {\em Applied and Computational Harmonic Analysis}, vol.~32,
  no.~2, pp.~268--279, 2012.

\bibitem{gong2014image}
Z.~Gong, Z.~Shen, and K.-C. Toh, ``Image restoration with mixed or unknown
  noises,'' {\em Multiscale Modeling \& Simulation}, vol.~12, no.~2,
  pp.~458--487, 2014.

\bibitem{quan2014data}
Y.~Quan, H.~Ji, and Z.~Shen, ``Data-driven multi-scale non-local wavelet frame
  construction and image recovery,'' {\em Journal of Scientific Computing},
  pp.~1--23, 2014.

\bibitem{bao2014convergence}
C.~Bao, H.~Ji, and Z.~Shen, ``Convergence analysis for iterative data-driven
  tight frame construction scheme,'' {\em Applied and Computational Harmonic
  Analysis}, 2014.

\bibitem{tai2015adaptive}
C.~Tai and W.~E, ``Multiscale adaptive representation of signals: I. the basic
  framework,'' {\em Preprint}, 2015.

\bibitem{ji2011robust}
H.~Ji, S.~Huang, Z.~Shen, and Y.~Xu, ``Robust video restoration by joint sparse
  and low rank matrix approximation,'' {\em SIAM Journal on Imaging Sciences},
  vol.~4, no.~4, pp.~1122--1142, 2011.

\bibitem{Dong2010Seg}
B.~Dong, A.~Chien, and Z.~Shen, ``Frame based segmentation for medical
  images,'' {\em Communications in Mathematical Sciences}, vol.~9(2),
  pp.~551--559, 2010.

\bibitem{tai2013wavelet}
C.~Tai, X.~Zhang, and Z.~Shen, ``Wavelet frame based multiphase image
  segmentation,'' {\em SIAM Journal on Imaging Sciences}, vol.~6, no.~4,
  pp.~2521--2546, 2013.

\bibitem{wendt2009wavelet}
H.~Wendt, P.~Abry, S.~Jaffard, H.~Ji, and Z.~Shen, ``Wavelet leader
  multifractal analysis for texture classification,'' in {\em Image Processing
  (ICIP), 2009 16th IEEE International Conference on}, pp.~3829--3832, IEEE,
  2009.

\bibitem{bao2014l0}
C.~Bao, H.~Ji, Y.~Quan, and Z.~Shen, ``$\ell_0$ norm based dictionary learning
  by proximal methods with global convergence,'' in {\em Computer Vision and
  Pattern Recognition (CVPR), 2014 IEEE Conference on}, pp.~3858--3865, IEEE,
  2014.

\bibitem{jiang2011highly}
Q.~Jiang and D.~K. Pounds, ``Highly symmetric bi-frames for triangle surface
  multiresolution processing,'' {\em Applied and Computational Harmonic
  Analysis}, vol.~31, no.~3, pp.~370--391, 2011.

\bibitem{dong2015surf}
B.~Dong, Q.~Jiang, C.~Liu, and Z.~Shen, ``Multiscale representation of surfaces
  by tight wavelet frames with applications to denoising,'' {\em Applied and
  Computational Harmonic Analysis, doi:10.1016/j.acha.2015.03.005}, 2015.

\bibitem{hammond2011wavelets}
D.~K. Hammond, P.~Vandergheynst, and R.~Gribonval, ``Wavelets on graphs via
  spectral graph theory,'' {\em Applied and Computational Harmonic Analysis},
  vol.~30, no.~2, pp.~129--150, 2011.

\bibitem{gavish2010multiscale}
M.~Gavish, B.~Nadler, and R.~R. Coifman, ``Multiscale wavelets on trees, graphs
  and high dimensional data: Theory and applications to semi supervised
  learning,'' in {\em Proceedings of the 27th International Conference on
  Machine Learning (ICML-10)}, pp.~367--374, 2010.

\bibitem{leonardi2013tight}
N.~Leonardi and D.~Van De~Ville, ``Tight wavelet frames on multislice graphs,''
  {\em Signal Processing, IEEE Transactions on}, vol.~61, no.~13,
  pp.~3357--3367, 2013.

\bibitem{dong2015sparse}
B.~Dong, ``Sparse representation on graphs by tight wavelet frames and
  applications,'' {\em Applied and Computational Harmonic Analysis,
  doi:10.1016/j.acha.2015.09.005}, 2015.

\bibitem{sapiro2001geometric}
G.~Sapiro, {\em Geometric partial differential equations and image analysis}.
\newblock Cambridge University Press, 2001.

\bibitem{OF}
S.~Osher and R.~Fedkiw, {\em {Level set methods and dynamic implicit
  surfaces}}.
\newblock Springer, 2003.

\bibitem{ChanShen}
T.~Chan and J.~Shen, {\em Image processing and analysis: variational, PDE,
  wavelet, and stochastic methods}.
\newblock Society for Industrial Mathematics, 2005.

\bibitem{ROF}
L.~Rudin, S.~Osher, and E.~Fatemi, ``{Nonlinear total variation based noise
  removal algorithms},'' {\em Phys. D}, vol.~60, pp.~259--268, 1992.

\bibitem{PM}
P.~Perona and J.~Malik, ``{Scale-space and edge detection using anisotropic
  diffusion},'' {\em IEEE Transactions on Pattern Analysis and Machine
  Intelligence}, vol.~12, no.~7, pp.~629--639, 1990.

\bibitem{DJS2013}
B.~Dong, Q.~Jiang, and Z.~Shen, ``{Image Restoration: Wavelet Frame Shrinkage,
  Nonlinear Evolution PDEs, and Beyond},'' {\em UCLA CAM Report}, vol.~13-78,
  2013.

\bibitem{CDOS2011}
J.~Cai, B.~Dong, S.~Osher, and Z.~Shen, ``Image restorations: total variation,
  wavelet frames and beyond,'' {\em Journal of American Mathematical Society},
  vol.~25(4), pp.~1033--1089, 2012.

\bibitem{CDS2014}
J.~Cai, B.~Dong, and Z.~Shen, ``Image restorations: a wavelet frame based model
  for piecewise smooth functions and beyond,'' {\em Applied and Computational
  Harmonic Analysis}, 2015.
\newblock http://dx.doi.org/10.1016/j.acha.2015.06.009.

\bibitem{mumford1989optimal}
D.~Mumford and J.~Shah, ``Optimal approximations by piecewise smooth functions
  and associated variational problems,'' {\em Communications on pure and
  applied mathematics}, vol.~42, no.~5, pp.~577--685, 1989.

\bibitem{Osher1990}
S.~Osher and L.~Rudin, ``Feature-oriented image enhancement using shock
  filters,'' {\em SIAM Journal on Numerical Analysis}, vol.~27, pp.~919--940,
  Aug 1990.

\bibitem{jia2011gpu}
X.~Jia, B.~Dong, Y.~Lou, and S.~Jiang, ``{GPU-based iterative cone-beam CT
  reconstruction using tight frame regularization},'' {\em Physics in Medicine
  and Biology}, vol.~56, pp.~3787--3807, 2011.

\bibitem{gao2011robust}
H.~Gao, J.-F. Cai, Z.~Shen, and H.~Zhao, ``Robust principal component
  analysis-based four-dimensional computed tomography,'' {\em Physics in
  medicine and biology}, vol.~56, no.~11, p.~3181, 2011.

\bibitem{gao20124d}
H.~Gao, R.~Li, Y.~Lin, and L.~Xing, ``4d cone beam ct via spatiotemporal tensor
  framelet,'' {\em Medical physics}, vol.~39, no.~11, pp.~6943--6946, 2012.

\bibitem{cai2014cine}
J.~Cai, X.~Jia, H.~Gao, S.~Jiang, Z.~Shen, and H.~Zhao, ``Cine cone beam ct
  reconstruction using low-rank matrix factorization: Algorithm and a
  proof-of-principle study.,'' {\em IEEE transactions on medical imaging},
  vol.~33, no.~8, pp.~1581--1591, 2014.

\bibitem{zhao2013tight}
B.~Zhao, H.~Gao, H.~Ding, and S.~Molloi, ``Tight-frame based iterative image
  reconstruction for spectral breast ct,'' {\em Medical physics}, vol.~40,
  no.~3, p.~031905, 2013.

\bibitem{sidky2006accurate}
E.~Sidky, C.~Kao, and X.~Pan, ``{Accurate image reconstruction from few-views
  and limited-angle data in divergent-beam CT},'' {\em Journal of X-Ray Science
  and Technology}, vol.~14, no.~2, pp.~119--139, 2006.

\bibitem{sidky2008image}
E.~Sidky and X.~Pan, ``{Image reconstruction in circular cone-beam computed
  tomography by constrained, total-variation minimization},'' {\em Physics in
  medicine and biology}, vol.~53, p.~4777, 2008.

\bibitem{Jia2010GPU}
X.~Y. Jia, R.~Li, W.~Y. Song, and S.~B. Jiang, ``Gpu-based fast cone beam ct
  reconstruction from undersampled and noisy projection data via total
  variation.,'' {\em Medical Physics}, vol.~37, no.~4, pp.~1757--1760, 2010.

\bibitem{zhang2010bregmanized}
X.~Zhang, M.~Burger, X.~Bresson, and S.~Osher, ``{Bregmanized nonlocal
  regularization for deconvolution and sparse reconstruction},'' {\em SIAM
  Journal on Imaging Sciences}, vol.~3, pp.~253--276, 2010.

\bibitem{chen2013limited}
Z.~Chen, X.~Jin, L.~Li, and G.~Wang, ``A limited-angle ct reconstruction method
  based on anisotropic tv minimization,'' {\em Physics in medicine and
  biology}, vol.~58, no.~7, p.~2119, 2013.

\bibitem{han2013image}
Z.~Han-Ming, W.~Lin-Yuan, Y.~Bin, L.~Lei, X.~Xiao-Qi, and L.~Li-Zhong, ``Image
  reconstruction based on total-variation minimization and alternating
  direction method in linear scan computed tomography,'' {\em Chinese Physics
  B}, vol.~22, no.~7, p.~078701, 2013.

\bibitem{thibault2007three}
J.-B. Thibault, K.~D. Sauer, C.~A. Bouman, and J.~Hsieh, ``A three-dimensional
  statistical approach to improved image quality for multislice helical ct,''
  {\em Medical physics}, vol.~34, no.~11, pp.~4526--4544, 2007.

\bibitem{wang2009iterative}
J.~Wang, T.~Li, and L.~Xing, ``Iterative image reconstruction for cbct using
  edge-preserving prior,'' {\em Medical physics}, vol.~36, no.~1, pp.~252--260,
  2009.

\bibitem{tang2009performance}
J.~Tang, B.~E. Nett, and G.-H. Chen, ``Performance comparison between total
  variation (tv)-based compressed sensing and statistical iterative
  reconstruction algorithms,'' {\em Physics in medicine and biology}, vol.~54,
  no.~19, p.~5781, 2009.

\bibitem{ouyang2011effects}
L.~Ouyang, T.~Solberg, and J.~Wang, ``Effects of the penalty on the penalized
  weighted least-squares image reconstruction for low-dose cbct,'' {\em Physics
  in medicine and biology}, vol.~56, no.~17, p.~5535, 2011.

\bibitem{lee2012improved}
H.~Lee, L.~Xing, R.~Davidi, R.~Li, J.~Qian, and R.~Lee, ``Improved compressed
  sensing-based cone-beam ct reconstruction using adaptive prior image
  constraints,'' {\em Physics in medicine and biology}, vol.~57, no.~8,
  p.~2287, 2012.

\bibitem{lauzier2013characterization}
P.~T. Lauzier and G.-H. Chen, ``Characterization of statistical prior image
  constrained compressed sensing (piccs): Ii. application to dose reduction,''
  {\em Medical physics}, vol.~40, no.~2, p.~021902, 2013.

\bibitem{xu2012low}
Q.~Xu, H.~Yu, X.~Mou, L.~Zhang, J.~Hsieh, and G.~Wang, ``Low-dose x-ray ct
  reconstruction via dictionary learning,'' {\em Medical Imaging, IEEE
  Transactions on}, vol.~31, no.~9, pp.~1682--1697, 2012.

\bibitem{chen2013improving}
Y.~Chen, X.~Yin, L.~Shi, H.~Shu, L.~Luo, J.-L. Coatrieux, and C.~Toumoulin,
  ``Improving abdomen tumor low-dose ct images using a fast dictionary learning
  based processing,'' {\em Physics in medicine and biology}, vol.~58, no.~16,
  p.~5803, 2013.

\bibitem{zhou2013adaptive}
W.~Zhou, J.-F. Cai, and H.~Gao, ``Adaptive tight frame based medical image
  reconstruction: a proof-of-concept study for computed tomography,'' {\em
  Inverse problems}, vol.~29, no.~12, p.~125006, 2013.

\bibitem{chen2014artifact}
Y.~Chen, L.~Shi, Q.~Feng, J.~Yang, H.~Shu, L.~Luo, J.-L. Coatrieux, and
  W.~Chen, ``Artifact suppressed dictionary learning for low-dose ct image
  processing,'' {\em Medical Imaging, IEEE Transactions on}, vol.~33, no.~12,
  pp.~2271--2292, 2014.

\bibitem{ron1997affine}
A.~Ron and Z.~Shen, ``Affine systems in {$L_2(\mathbb{R}^d)$}: The analysis of
  the analysis operator,'' {\em Journal of Functional Analysis}, vol.~148,
  no.~2, pp.~408--447, 1997.

\bibitem{ron1997affineII}
A.~Ron and Z.~Shen, ``Affine systems in {$L_2(\mathbb{R}^d)$} ii: dual
  systems,'' {\em Journal of Fourier Analysis and Applications}, vol.~3, no.~5,
  pp.~617--638, 1997.

\bibitem{Dau}
I.~Daubechies, {\em {Ten lectures on wavelets}}, vol.~CBMS-NSF Lecture Notes,
  SIAM, nr. 61.
\newblock Society for Industrial and Applied Mathematics, 1992.

\bibitem{Daubechies2003}
I.~Daubechies, B.~Han, A.~Ron, and Z.~Shen, ``{Framelets: MRA-based
  constructions of wavelet frames},'' {\em Applied and Computational Harmonic
  Analysis}, vol.~14, pp.~1--46, Jan 2003.

\bibitem{ShenICM2010}
Z.~Shen, ``Wavelet frames and image restorations,'' in {\em Proceedings of the
  International Congress of Mathematicians}, vol.~4, pp.~2834--2863, 2010.

\bibitem{Dong2010IASNotes}
B.~Dong and Z.~Shen, ``{MRA-Based Wavelet Frames and Applications},'' {\em IAS
  Lecture Notes Series, Summer Program on ``The Mathematics of Image
  Processing", Park City Mathematics Institute}, 2010.

\bibitem{aharon2006k}
M.~Aharon, M.~Elad, and A.~Bruckstein, ``K-svd: An algorithm for designing
  overcomplete dictionaries for sparse representation,'' {\em IEEE TRANSACTIONS
  ON SIGNAL PROCESSING}, vol.~54, no.~11, p.~4311, 2006.

\bibitem{elad2005simultaneous}
M.~Elad, J.~Starck, P.~Querre, and D.~Donoho, ``{Simultaneous cartoon and
  texture image inpainting using morphological component analysis (MCA)},''
  {\em Applied and Computational Harmonic Analysis}, vol.~19, no.~3,
  pp.~340--358, 2005.

\bibitem{starck2005image}
J.~Starck, M.~Elad, and D.~Donoho, ``{Image decomposition via the combination
  of sparse representations and a variational approach},'' {\em IEEE
  transactions on image processing}, vol.~14, no.~10, pp.~1570--1582, 2005.

\bibitem{Attouch2010Proximal}
H.~Attouch, J.~Bolte, P.~Redont, and A.~Soubeyran, ``Proximal alternating
  minimization and projection methods for nonconvex problems: An approach based
  on the kurdyka-?ojasiewicz inequality,'' {\em Mathematics of Operations
  Research}, vol.~35, no.~2, pp.~438--457, 2010.

\bibitem{bolte2014proximal}
J.~Bolte, S.~Sabach, and M.~Teboulle, ``Proximal alternating linearized
  minimization for nonconvex and nonsmooth problems,'' {\em Mathematical
  Programming}, vol.~146, no.~1-2, pp.~459--494, 2014.

\bibitem{xu2014globally}
Y.~Xu and W.~Yin, ``A globally convergent algorithm for nonconvex optimization
  based on block coordinate update,'' {\em arXiv preprint arXiv:1410.1386},
  2014.

\bibitem{kurdyka1998}
K.~Kurdyka, ``On gradients of functions denable in o-minimal structures,''
  {\em Annales de l'institut Fourier}, vol.~48, pp.~769--784, 1998.

\bibitem{lojasiewicz1993geometrie}
S.~{\L}ojasiewicz, ``Sur la g{\'e}om{\'e}trie semi-et sous-analytique,''
  vol.~43, no.~5, pp.~1575--1595, 1993.

\bibitem{Segars2000Development}
W.~P. Segars, D.~S. Lalush, and B.~M.~W. Tsui, ``Development of an interactive
  software application to model patient populations in the 4d nurbs-based
  cardiac torso phantom,'' in {\em Nuclear Science Symposium Conference Record,
  2000 IEEE}, pp.~20/51--20/55 vol.3, 2000.

\end{thebibliography}
\end{document}